\documentclass{article}
\usepackage{graphicx}
\usepackage{amsmath, amsthm, amssymb}
\usepackage[bottom]{footmisc}
\usepackage{cite}
\usepackage{longtable}
\graphicspath{{figures/}}
\usepackage{multirow}
\usepackage{subfigure}
\usepackage{algorithm}
\usepackage{algorithmic}
\usepackage[dvipsnames]{xcolor}
\usepackage{tabularx}

\newlength{\totalhmargin}
\newlength{\totalvmargin}
\def\R{\mathbb{R}}

\theoremstyle{definition}
\newtheorem{remark}{Remark}
\newtheorem{lemma}{Lemma}

\newtheorem{theorem}{Theorem}

\title{A Rigorous Mathematical Theory for Topological Phases and Edge Modes in Spring-mass Mechanical Systems}

\author{Ridvan Ozdemir
  \thanks{Department of Mathematics and Statistics, Auburn University, Auburn, AL 36849.  \tt rzo0012@auburn.edu.} \,\, and \, Junshan Lin
  \thanks{Department of Mathematics and Statistics, Auburn University, Auburn, AL 36849.  \tt jzl0097@auburn.edu.} }
  
 \date{}
 
\begin{document}
\newpage
\maketitle
\begin{abstract}
In this work, we examine the topological phases of the spring-mass lattices when the spatial inversion symmetry of the system is broken and prove the existence of edge modes when two  lattices with different topological phases are glued together. In particular, for the one-dimensional lattice consisting of an infinite array of masses connected by springs, we show that the Zak phase of the lattice is quantized, only taking the value $0$ or $\pi$. We also prove the existence of an edge mode when two semi-infinite lattices with distinct Zak phases are connected. For the two-dimensional honeycomb lattice, we characterize the valley Chern numbers of the lattice when the masses on the lattice vertices are uneven. The existence of edge modes is proved for a joint honeycomb lattice formed by gluing two semi-infinite lattices with opposite valley Chern numbers together.
\end{abstract}

\section{Introduction}
The recent development of topological insulators in condensed matter physics has opened up new avenues for localization and confinement of classical waves. In topological insulators, 
an insulating bulk electronic material can support localized edge states on its surface, and the existence of edge states is associated with the topological invariant of the bulk electron material \cite{kane2013topological}. The extension of concepts in topological insulators to classical waves was proposed in the seminal work \cite{raghu2008analogs}, where the topological phases in the electromagnetic wave systems were introduced using the wave functions in the momentum space. Since then extensive research has been devoted to control acoustic, electromagnetic and mechanical waves in the same way as solids modulating electrons in topological insulators \cite{khanikaev2013photonic,  lu2014topological, ma2019topological, ozawa2019topological, xiao2015geometric}.

There exist mainly two strategies to realize topological wave insulators for classical waves. The first strategy mimics the so-called quantum Hall effect in topological insulator using active components to break the time-reversal symmetry of the system. This is realized by moderating rotational motion of air  in acoustic media or applying the external magnetic field in electromagnetic media \cite{khanikaev2015topologically, nash2015topological,wang2009observation}.
The second strategy relies on an analogue of the quantum spin Hall effect or quantum valley Hall effect, and it uses passive components to break the inversion symmetry of the system \cite{lu2016valley,  ye2017observation, ma2016all,wu2015scheme}. In this work, we investigate the spring-mass topological mechanical systems using the second strategy. The inversion symmetry in each periodic cell of the system is broken by tuning either the mass parameter or the spring constant.
The setup of the topological mechanical material was introduced in \cite{pal2019topologically}, and our goal in this work is to provide a rigorous mathematical theory for the topological phases and edge modes in such mechanical systems. The spring-mass topological mechanical systems using the first strategy was realized in \cite{zheng2022toward}. The mathematical studies for the corresponding topological phases and edge modes will be forthcoming.

We examine topological mechanical systems in one and two dimensions. The periodic lattice in one dimension is constructed over the real line with identical masses, wherein each mass is connected by two springs with different spring constants. We derive the Zak phase of the lattice and show that its value is quantized when the spring constant varies, only taking the value $0$ or $\pi$. Additionally, we prove the existence of edge modes when two semi-infinite mechanical systems with different Zak phases are joined together.  In two dimensions, the periodic mechanical system is constructed over a honeycomb lattice, wherein each periodic cell consists of two different masses that are connected to the neighboring masses with the identical springs. We investigate the valley Chern number and examine how its value is related to the change of masses. Furthermore, we prove the existence of the edge modes in a joint mechanical system formed by two honeycomb lattices with opposite valley Chern numbers. We would like to refer the readers to the mathematical studies of edge modes in acoustic and electromagnetic waves in \cite{ammari2020topologically, ammari2022robust, drouot2020edge, fefferman2016edge, fefferman2014topologically, qiu2023mathematical, zhang2021superior}. In general, the number of edge modes is equal to the difference of the bulk topological invariants across the interface, which is known as the bulk-edge correspondence. We refer to \cite{bal2017topological, bal2019continuous, drouot2021microlocal,  elbau2002equality, elgart2005equality} 
for the bulk-edge correspondence in electron models for topological insulators and several elliptic partial differential equation models.

The rest of the paper is organized as follows. In Section \ref{1d}, we consider the one-dimensional spring-mass mechanical system. The Zak phase of the lattice is given in Lemma \ref{Zak_lemma} and the existence of the edge modes for the joint lattice is established in Theorem \ref{one_d_main_thm}. In Section \ref{2d}, we investigate the two-dimensional mechanical system over the honeycomb lattice. The valley Chern number for the lattice is summarized in Lemma \ref{chern_lemma} when the masses on the lattice vertices are uneven. Finally, the existence of edge modes for the joint topological mechanical insulator is established in Theorem \ref{two_d_main_thm}.

\section{One-dimensional Topological Mechanical Systems}\label{1d}

\subsection{Periodic Mechanical System}\label{zaksection}

\begin{figure}[h]
\centering
\includegraphics[scale=0.5]{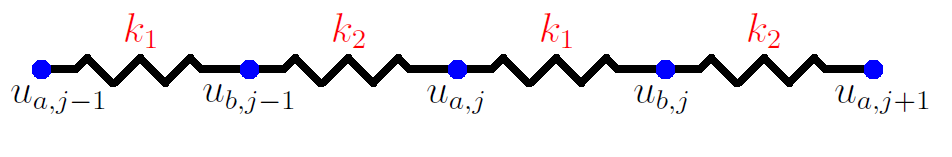}
\caption{One-dimensional lattice consisting of an array of masses connected by springs.}
\label{gencase}
\end{figure}

We consider the one-dimensional periodic mechanical system shown in Figure~\ref{gencase}, wherein an infinite array of masses are arranged along real line. The spring connecting two masses in the unit cell $j$ and two masses between the cell $j$ and $j+1$ attains the spring constants
\begin{align}\label{1d_ks}
k_1 = k(1+\gamma)\,\,\,\text{and}\,\,\, k_2 = k(1-\gamma),
\end{align}
where $\gamma$ is a stiffness parameter and $k$ is the mean stiffness of the springs. The displacements of masses in the unit cell $j$ satisfies the following equations:
\begin{align*}
mU_{a,j}'' + k_{1}\left(U_{a,j}-U_{b,j}\right) + k_{2}\left(U_{a,j}-U_{b,j-1} \right) &= 0,\\
mU_{b,j}'' + k_{2}\left(U_{b,j}-U_{a,j+1}\right) + k_{1}\left(U_{b,j}-U_{a,j} \right) &= 0.
\end{align*}
We consider the solution in the form of 
\begin{align}\label{soln_zak}
U_{a,j}(t) = u_{a}e^{i\omega \tau+\mu j} \,\, \text{and} \,\, U_{b,j}(t) = u_{b}e^{i\omega \tau + \mu j}
\end{align}
where $u_a$ and $u_b$ are the amplitudes of the displacements of masses, $j$ denotes the cell index, $\omega$ is the frequency, $\tau = \sqrt{k/m}t$ is nondimensional time scale and $\mu \in [-\pi,\pi]$ is the nondimensional wave number. Then $u_{a}$ and $u_{b}$ satisfy
\begin{align*}
-m\frac{k}{m}\omega^2 u_a + k(1+\gamma) (u_a - u_b)e^{i\omega \tau+\mu j} + k(1-\gamma)(u_a-u_be^{-i\mu})e^{i\omega \tau+\mu j} &= 0,\\
-m\frac{k}{m}\omega^2 u_b + k(1+\gamma) (u_b - u_ae^{i\mu})e^{i\omega \tau+\mu j} + k(1-\gamma)(u_b-u_a)e^{e^{i\omega \tau+\mu j}} &= 0,
\end{align*}
which reduces to the eigenvalue problem
\begin{align}\label{geneigen}
\begin{bmatrix}
2 & \bar{a}(\mu) \\
a(\mu) & 2
\end{bmatrix}
\begin{bmatrix}
u_a \\
u_b
\end{bmatrix}
= \omega^2
\begin{bmatrix}
u_a \\
u_b
\end{bmatrix},
\end{align}
where $a(\mu) = -\left(1+\gamma\right) - \left(1-\gamma\right)e^{i\mu}$ and $\bar{a}(\mu)$ is the complex conjugate of $a(\mu)$. The eigenpairs of matrix in (\ref{geneigen}) are
\begin{align*}
\lambda_{\pm}\left(\mu\right) = 2 \pm |a(\mu)|\,\,\,\text{and}\,\,\,\textbf{v}_\pm(\mu) = 
\frac{1}{\sqrt{2}}\begin{bmatrix}
\frac{\bar{a}(\mu)}{\pm|a(\mu)|} \\
1
\end{bmatrix},
\end{align*}
with $||\textbf{v}_\pm(\mu)||_2=1$. We note that if $\gamma \neq 0$, then $|a(\mu)| \neq 0$ and there is a gap between two bands $\lambda_-(\mu)$ and $\lambda_+(\mu)$ for $\mu \in [-\pi,\pi]$. We call this gap as the band gap interval
\begin{align}\label{bandgap_generic}
I(\gamma) := (\sqrt{2(1-|\gamma|)}, \sqrt{2(1+|\gamma|)}),
\end{align}
where $\sqrt{2(1-|\gamma|)}$ and $\sqrt{2(1+|\gamma|)}$ are maximum and minimum values of $\sqrt{\lambda_-(\mu)}$ and $\sqrt{\lambda_+(\mu)}$ respectively. We investigate the dynamics of the system for the frequency $\omega$ located in the band gap $I\left(\gamma\right)$ which is induced by a topological index called Zak phase.

The Zak phase associated with the frequency band $\lambda(\mu)$ is defined by (cf. \cite{vanderbilt}) 
\begin{align}\label{zak_phase}
\theta^{Zak} = \int_{-\pi}^\pi \left[i \left(\textbf{v}(\mu)\right)^H \partial_{\mu}\textbf{v}(\mu)\right]d\mu = -Im\left(\int_{-\pi}^\pi \left[\left(\textbf{v}(\mu)\right)^H \partial_{\mu}\textbf{v}(\mu)\right]d\mu\right),
\end{align}
where $\textbf{v} = \textbf{v}_+$ or $\textbf{v} = \textbf{v}_-$ is the eigenvector associated with the eigenvalue $\lambda_{\pm}(\mu)$ of the matrix defined in (\ref{geneigen}) and $\textbf{v}^H$ stands for complex conjugate transpose of $\textbf{v}$. To avoid the difficulties in calculation of the composition of differentiation and integration, we use the discretization of the integral in (\ref{zak_phase}). To this purpose, for $\mu_n = n\pi/N$, $n = -N,-(N-1),...,N-1,N$ where $N \in \mathbb{Z}^+$, we observe that
\begin{align*}
\log\left[ \left(\textbf{v}(\mu_n)\right)^H \textbf{v}(\mu_{n+1})\right] &= \log\left[\left(\textbf{v}(\mu_n)\right)^H \left(\textbf{v}(\mu_{n}) + \partial_\mu \textbf{v}(\mu_{n}) (\mu_{n+1}-\mu_n) \right) + O\left(N^{-2}\right)\right]\\
&=\log\left[\left(\textbf{v}(\mu_n)\right)^H \textbf{v}(\mu_{n}) + \left(\textbf{v}(\mu_n)\right)^H \partial_\mu \textbf{v}(\mu_{n}) (\mu_{n+1}-\mu_n) + O\left(N^{-2}\right)\right] \\
&= \log\left[1 + \left(\textbf{v}(\mu_n)\right)^H \partial_\mu \textbf{v}(\mu_{n}) (\mu_{n+1}-\mu_n)+ O\left(N^{-2}\right) \right]\\
&=\left(\textbf{v}(\mu_n)\right)^H \partial_\mu \textbf{v}(\mu_{n}) (\mu_{n+1}-\mu_n) + O(N^{-2}).
\end{align*}
Then, by discretization of the integral, Zak phase can be written as
\begin{align*}
\theta^{Zak} &= -Im\left(\int_{-\pi}^\pi \left[ \left(\textbf{v}(\mu)\right)^H \partial_{\mu}\textbf{v}(\mu)\right]d\mu\right) \\
&= -Im\left(\lim_{N \to \infty} \sum_{n=-N}^{N-1}\left(\textbf{v}(\mu_n)\right)^H \partial_\mu \textbf{v}(\mu_{n}) (\mu_{n+1}-\mu_n)\right) \\
&=\lim_{N \to \infty}  -\sum_{n=-N}^{N-1}Im\left(\log\left[ \left(\textbf{v}(\mu_n)\right)^H \textbf{v}(\mu_{n+1})\right]\right).
\end{align*}
We define the discrete Zak Phase as
\begin{align*}
\theta^{Zak}_N := - \sum_{n = -N}^{N-1} Im\left(\log\left[\left(\textbf{v}\left(\mu_n\right)\right)^H\textbf{v}\left(\mu_{n+1}\right)\right]\right).
\end{align*}
We have the following lemma for a complex number:
\begin{lemma}\label{arg}
For a complex number $z = re^{i\theta}$ with $\theta \in [-\pi,\pi]$, if $z+1 = |z+1|e^{i\beta}$, there holds 
\begin{align*}
\begin{cases}
\beta < \theta/2,\,\,\, \text{for}\,\,\,r<1,\\
\beta = \theta/2,\,\,\, \text{for}\,\,\,r=1,\\
\beta > \theta/2,\,\,\, \text{for}\,\,\,r>1.\\
\end{cases}
\end{align*}
\end{lemma}

\begin{proof}
For $z = a+ib$, the half angle formula gives that
\begin{align*}
\tan(\theta/2) = \frac{\sin\left(\theta\right)}{1+\cos\left(\theta\right)} = \frac{b/r}{1+a/r} = \frac{b}{r+a}.
\end{align*}
Also, since $\tan(\beta) = \frac{b}{1+a}$, the results follow by the fact that tangent is an increasing function on $[-\pi,\pi]$.

\end{proof}

\begin{figure}[h]
\centering     
\subfigure[$\gamma > 0$]{\label{betapos}\includegraphics[width=73mm]{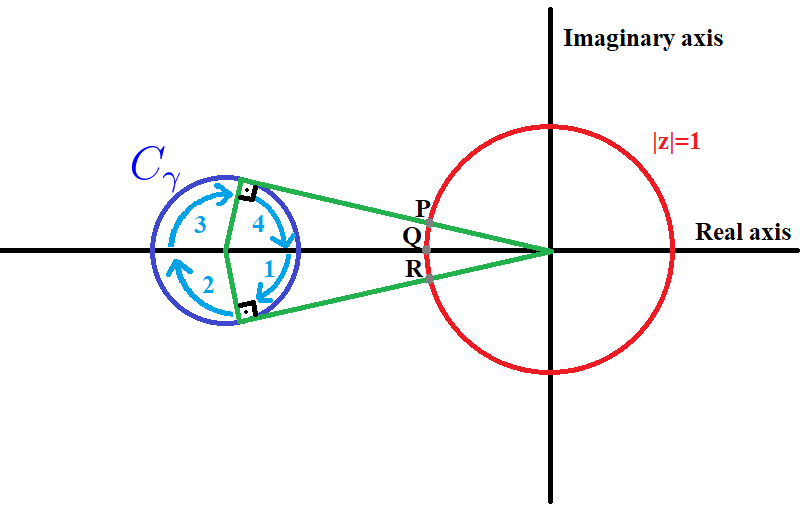}}
\subfigure[$\gamma < 0$]{\label{betaneg}\includegraphics[width=73mm]{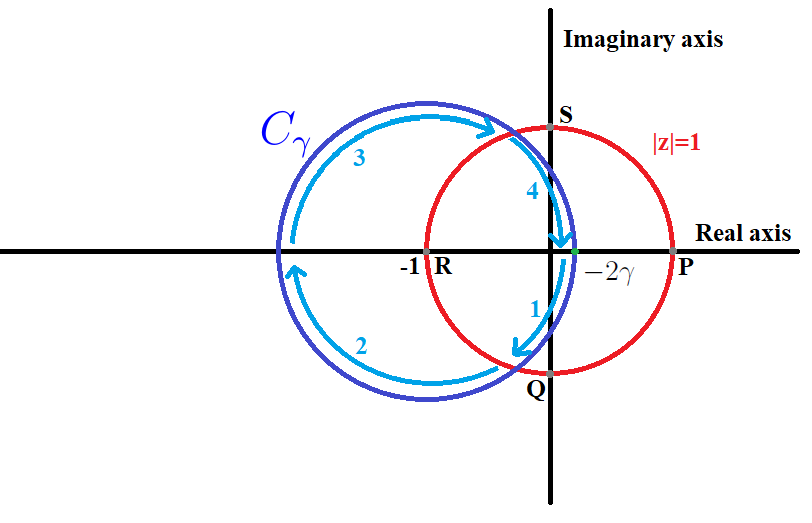}}
\caption{When $\gamma > 0$, $C_{\gamma}$, the circle with center $-(1+\gamma)$ and radius $1-\gamma$, does not enclose the origin. As $a(\mu)$ completes one turn on $C_{\gamma}$ from 1 to 4, the trajectory of $\frac{a(\mu)}{|a(\mu)|}$ is the path $Q \to R \to Q \to P \to Q$ on the unit circle and its argument oscillates near $\pi$. When $\gamma < 0$, $C_{\gamma}$ contains origin. As $a(\mu)$ completes one turn on $C_{\gamma}$, the trajectory of $\frac{a(\mu)}{|a(\mu)|}$ is the path $P \to Q \to R \to S \to P$ on the unit circle and its argument goes from $2\pi$ to $0$.}
\label{zak_graph}
\end{figure}
\begin{lemma}\label{Zak_lemma}
For the Zak phase $\theta^{Zak}$ associated with the band $\lambda_-(\mu)$ or $\lambda_+(\mu)$ of the system (\ref{geneigen}), we have
\begin{align*}
\begin{cases}
\theta^{Zak} = 0,\,\,\,\text{if}\,\,\,\gamma>0,\\
\theta^{Zak} = \pi,\,\,\,\text{if}\,\,\,\gamma<0.
\end{cases}
\end{align*}
\end{lemma}
\begin{proof}
For $\lambda_+(\mu)$, we have $\mathbf{v} = \mathbf{v}_+$. Let $\theta_n$ be the argument of the first component of $\textbf{v}_+(\mu_n)$. Then a direct computation leads to
\begin{align*}
\textbf{v}_+^H\left(\mu_n\right)\textbf{v}_+\left(\mu_{n+1}\right) &= \frac{1}{2}\left(\begin{bmatrix}
\frac{a\left(\mu_n\right)}{|a\left(\mu_n\right)|} &
1
\end{bmatrix}
\begin{bmatrix}
\frac{\bar{a}\left(\mu_{n+1}\right)}{|a\left(\mu_{n+1}\right)|} \\
1
\end{bmatrix}
\right) \\
&=\frac{1}{2} \left(e^{i(\theta_{n+1}-\theta_n)}+ 1 \right) \\
&= \frac{c_n}{2}\exp\left(i\frac{\theta_{n+1}-\theta_n}{2}\right),
\end{align*}
where $c_n$ is the modulus of $e^{i(\theta_{n+1}-\theta_n)}+ 1$ and we have used Lemma~\ref{arg} in the last step. Thus the discrete Zak phase
\begin{align*}
\theta_N^{Zak} = - \sum_{n = -N}^{N-1} Im\left(\log\left[\textbf{v}_+^H\left(\mu_{n}\right)\textbf{v}_+\left(\mu_{n+1}\right)\right]\right) &=- \sum_{n = -N}^{N-1} Im\left(\log \left[\frac{c_n}{2}\left(\exp\left(i\frac{\theta_{n+1}-\theta_n}{2}\right)\right)\right]\right)\\
&= -\frac{\theta_N - \theta_{-N}}{2}.
\end{align*}
If $\gamma>0$, we denote the circle with center $-(1+\gamma)$ and radius $1-\gamma$ as $C_\gamma$. As $\mu$ goes from $-\pi$ to $\pi$, $\bar{a}(\mu) = -(1+\gamma) - (1-\gamma)e^{-i\mu}$ completes one turn on the circle $C_\gamma$ clockwisely starting from the point $-2\gamma$ on complex plane. Since $|-(1+\gamma)| = 1+\gamma > 1-\gamma$, the distance between the center of $C_\gamma$ and the origin is greater than its radius, hence $C_\gamma$ does not enclose the origin (see Figure \ref{betapos}). Therefore, the argument of $\frac{a\left(\mu\right)}{|a\left(\mu\right)|}$ oscillates around $\pi$ as $\mu$ goes from $-\pi$ to $\pi$. As a result, we obtain $\theta_N = \theta_{-N}$ and
\begin{align*}
\theta_N^{Zak} = -\frac{\theta_N - \theta_{-N}}{2} = 0.
\end{align*}
If $\gamma<0$, $\bar{a}(\mu)$ still completes one turn on the circle $C_\gamma$. However, noting that the distance between the center of $C_\gamma$ and the origin is less than the radius of $C_\gamma$ and $C_\gamma$ encloses the origin. Thus the argument of $\frac{a\left(\mu\right)}{|a\left(\mu\right)|}$ goes from $2\pi$ to $0$ (see Figure \ref{betaneg}) and we have
\begin{align*}
\theta_N^{Zak} = -\frac{\theta_N - \theta_{-N}}{2} = -\frac{0 - 2\pi}{2} = \pi.
\end{align*}
For $\lambda_-(\mu)$, we can obtain the same results by similar calculations. Therefore, the proof is complete by noting that $\theta^{Zak} = \lim_{N \to \infty}\theta_N^{Zak}$.
\end{proof}

\subsection{Edge Modes for the Topological Mechanical System}

\begin{figure}[h]
\centering
\includegraphics[scale=0.5]{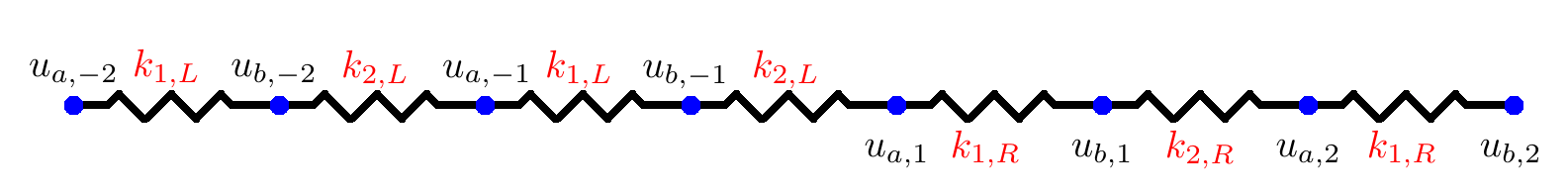}
\caption{The topological mechanical system in one dimension}
\label{system1d}
\end{figure}

We construct a joint system by gluing two periodic mechanical systems with different spring constants as shown in Figure \ref{system1d}. On the left, the spring constants defined in (\ref{1d_ks}) for each unit cell are $k_{1,L} = k(1+\gamma_L)$, $k_{2,L} = k(1-\gamma_L)$ and the spring constants for each unit cell on the right are $k_{1,R} = k(1+\gamma_R)$ and $k_{2,R} = k(1-\gamma_R)$. We assume that $\gamma_L$ and $\gamma_R$ have different signs. In the light of Lemma \ref{Zak_lemma}, these two systems attain different Zak phases, $\theta_L \neq \theta_R$. For $j = 2, 3,...$, the governing equations for the displacement of the masses at the left and right periodic systems are
\begin{equation} \label{generic_eqns_L}
\begin{aligned}
mU_{a,-j+1}'' + k_{1,L}\left(U_{a,-j+1}-U_{b,-j+1}\right) + k_{2,L}\left(U_{a,-j+1}-U_{b,-j} \right) &= 0,\\
mU_{b,-j}'' + k_{2,L}\left(U_{b,-j}-U_{a,-j+1}\right) + k_{1,L}\left(U_{b,-j}-U_{a,-j} \right) &= 0,
\end{aligned}
\end{equation}
and
\begin{equation} \label{generic_eqns_R}
\begin{aligned}
mU_{a,j}'' + k_{1,R}\left(U_{a,j}-U_{b,j}\right) + k_{2,R}\left(U_{a,j}-U_{b,j-1} \right) &= 0,\\
mU_{b,j-1}'' + k_{2,R}\left(U_{b,j-1}-U_{a,j}\right) + k_{1,R}\left(U_{b,j-1}-U_{a,j-1} \right) &= 0.
\end{aligned}
\end{equation}
respectively. The governing equations for the displacement of the masses located at the interface of two periodic systems are
\begin{equation} \label{1dmiddle_eqn}
\begin{aligned}
mU_{a,1}'' + k_{1,R}\left(U_{a,1}-U_{b,1}\right) + k_{2,L}\left(U_{a,1}-U_{b,-1}\right)&=0,\\
mU_{b,-1}'' + k_{2,L}\left(U_{b,-1} - U_{a,1}\right) + k_{1,L} \left(U_{b,-1} - U_{a,-1}\right) &=0.
\end{aligned}
\end{equation}

A non-trivial solution $(U_{a,j}(t),U_{b,j}(t))$ for (\ref{generic_eqns_L}) - (\ref{1dmiddle_eqn}) which decays to zero as $j \to \pm\infty$ is called an edge mode.

In the rest of this subsection, we aim to show the existence of edge modes for the joint system in Figure \ref{system1d} when $\theta_L \neq \theta_R$ and $\omega \in I$, the common band gap, where
\begin{align*}
I : = I(\gamma_R) \cap I(\gamma_L) = \left[\sqrt{2(1-\left|\gamma_R\right|)},\sqrt{2(1+\left|\gamma_R\right|)}\right] \cap \left[\sqrt{2(1-\left|\gamma_L\right|)},\sqrt{2(1+\left|\gamma_L\right|)}\right].
\end{align*}
Our main result is stated in the following theorem:

\begin{theorem}\label{one_d_main_thm}
(Existence of edge modes) If $\gamma_L\gamma_R<0$ such that two periodic mechanical systems with the Zak phase $\theta_L\neq \theta_R$ are glued together as shown in Figure \ref{system1d}, then there exists an edge mode $(U_{a,j}(t),U_{b,j}(t))$ in the form of 
\begin{align}\label{1dsoln}
U_{a,j}(t) = u_{a,j}e^{i\omega \tau} \,\, \text{and} \,\, U_{b,j}(t) = u_{b,j}e^{i\omega \tau}
\end{align}
where $j$ denotes the cell index, $\tau$ is the nondimensional time scale and $\omega \in I$.
\end{theorem}

\subsubsection{Transfer Matrix for the Periodic System}

Assume that the solutions of (\ref{generic_eqns_L})-(\ref{generic_eqns_R}) take the form in (\ref{1dsoln}). For the right periodic system, $u_{a,j}$ and $u_{b,j}$ satisfy
\begin{align*}
-m \frac{k}{m}\omega^2 u_{a,j}e^{i\omega \tau} + k(1+\gamma_R)\left(u_{a,j}e^{i\omega \tau} - u_{b,j}e^{i\omega \tau} \right) + k(1-\gamma_R)\left(u_{a,j}e^{i\omega \tau} - u_{b,j-1}e^{i\omega \tau} \right) &= 0,\\
-m \frac{k}{m}\omega^2 u_{b,j-1}e^{i\omega \tau} + k(1-\gamma_R)\left(u_{b,j-1}e^{i\omega \tau} - u_{a,j}e^{i\omega \tau} \right) + k(1+\gamma_R)\left(u_{b,j-1}e^{i\omega \tau} - u_{a,j-1}e^{i\omega \tau} \right) &= 0,
\end{align*}
which implies
\begin{align*}
\left(2-\omega^2\right) u_{a,j} - (1+\gamma_R)u_{b,j} - (1-\gamma_R)u_{b,j-1} &= 0,\\
\left(2-\omega^2\right) u_{b,j-1} - (1-\gamma_R)u_{a,j} - (1+\gamma_R)u_{a,j-1} &= 0.
\end{align*}
This can be written as the system $A\mathbf{u_{j-1}} = B\mathbf{u_j}$, where
\begin{align*}
A = \begin{bmatrix}
0 & 1-\gamma_R\\
-(1+\gamma_R) & 2-\omega^2
\end{bmatrix}
,\,\,
B = \begin{bmatrix}
2-\omega^2 & -(1+\gamma_R)\\
1-\gamma_R & 0
\end{bmatrix} \,\, \text{and} \,\,
\mathbf{u_j} = \begin{bmatrix}
u_{a,j}\\
u_{b,j}
\end{bmatrix}.
\end{align*}
We rewrite $A\mathbf{u}_{j-1} = B\mathbf{u}_j$ as
\begin{align*}
T_R\mathbf{u}_{j-1} = \mathbf{u}_j,
\end{align*}
where the transfer matrix
\begin{align*}
T_R = B^{-1}A = \begin{bmatrix}
-\frac{1+\gamma_R}{1-\gamma_R} & \frac{2-\omega^2}{1-\gamma_R}\\
-\frac{2-\omega^2}{1-\gamma_R} & \frac{(2-\omega^2)^2 - (1-\gamma_R)^2}{1-\gamma_R^2}
\end{bmatrix}.
\end{align*}
It can be shown that
\begin{align*}
det\left(T_R\right) = 1.
\end{align*}
The eigenvalues of $T_R$ are
\begin{align*}
\lambda_{\pm,R}&= \frac{1}{2}\left(\frac{\left(2-\omega^2\right)^2 - (1-\gamma_R)^2 - (1+\gamma_R)^2}{1-\gamma_R^2} \pm \sqrt{\frac{1}{\left(1-\gamma_R^2\right)^2} \left[\omega^4-4\omega^2\right]\left[\left(2-\omega^2\right)^2 - 4\gamma_R^2\right]}\right)\\
&= \frac{1}{2(1-\gamma_R^2)}\left(\left(2-\omega^2\right)^2 - (1-\gamma_R)^2 - (1+\gamma_R)^2 \pm\sqrt{\omega^2\left(\omega^2-4\right)\left[\left(2-\omega^2\right)^2-4\gamma_R^2\right]}\right).
\end{align*}
The corresponding eigenvectors $\mathbf{e}_{\pm,R}$ are
\begin{align}\label{one_d_evector_general}
\mathbf{e}_{\pm,R}= \begin{bmatrix}
(2-\omega^2)\\
1+\gamma_R+\left(1-\gamma_R\right)\lambda_{\pm,R}
\end{bmatrix}.
\end{align}
For the left periodic system, similar calculations give that
\begin{align*}
T_L \mathbf{u}_{-j+1} = \mathbf{u}_{-j}
\end{align*}
where the transfer matrix
\begin{align*}
T_L = \begin{bmatrix}
\frac{(2-\omega^2)^2 - \left(1-\gamma_L\right)^2}{1-\gamma_L^2} & -\frac{2-\omega^2}{1-\gamma_L}\\
\frac{2-\omega^2}{1-\gamma_L} & \frac{1+\gamma_L}{1-\gamma_L}
\end{bmatrix},
\end{align*}
with eigenvalues
\begin{align*}
\lambda_{\pm,L}:= \frac{1}{2(1-\gamma_L^2)}\left(\left(2-\omega^2\right)^2 - (1-\gamma_L)^2 - (1+\gamma_L)^2 \pm \sqrt{\omega^2\left(\omega^2-4\right)\left[\left(2-\omega^2\right)^2-4\gamma_L^2\right]}\right)
\end{align*}
and the corresponding eigenvector
\begin{align*}
\mathbf{e}_{\pm,L} = \begin{bmatrix}
1+\gamma_L + \left(1-\gamma_L\right)\lambda_{\pm,L}\\
2-\omega^2
\end{bmatrix}.
\end{align*}
Note that $\sqrt{2(1-\left|\gamma_L\right|)} < \sqrt{2(1+\left|\gamma_R\right|)} \leq \sqrt{2(1+\left|\gamma_L\right|)}$ and $\sqrt{2(1-\left|\gamma_R\right|)} < \sqrt{2(1+\left|\gamma_L\right|)} \leq \sqrt{2(1+\left|\gamma_R\right|)}$ when $|\gamma_R| \leq |\gamma_L|$ and $|\gamma_L| \leq |\gamma_R|$ respectively. It follows that $I\neq \emptyset$.

For $\omega \in I$, $\lambda_{\pm,R}$ and $\lambda_{\pm,L}$ are real. Since $\det(T_R) = 1$ and $\lambda_{+,R} \neq \lambda_{-,R}$, one of $|\lambda_{+,R}|$ and $|\lambda_{-,R}|$ is less than $1$, $\mathbf{e}_{+,R}$ and $\mathbf{e}_{-,R}$ are linearly independent and $\{\mathbf{e}_{+,R}, \mathbf{e}_{-,L}\}$ form a basis of $\R^2$. Hence, $\mathbf{u}_1 = \begin{bmatrix} u_{a,1}\\ u_{b,1} \end{bmatrix}$ can be written as
\begin{align*}
\mathbf{u}_1 = a_1 \mathbf{e}_{+,R} + a_2 \mathbf{e}_{-,R}
\end{align*} 
for some constants $a_1, a_2 \in \R$. Then
\begin{align*}
\mathbf{u}_j = T_R^j \mathbf{u}_1 = a_1 \left(\lambda_{+,R}\right)^{j-1}\mathbf{e}_{+,R} + a_2 \left(\lambda_{-,R}\right)^{j-1}\mathbf{e}_{-,R}.
\end{align*}
For $\mathbf{u}_j$ to vanish as $j \to \infty$, it is necessary that $\mathbf{u}_1 = \begin{bmatrix} u_{a,1}\\ u_{b,1} \end{bmatrix}$ is parallel to the eigenvector of $T_R$ whose corresponding eigenvalue has absolute value less than 1. Similarly, $\mathbf{u}_{-1}=\begin{bmatrix} u_{a,-1}\\ u_{b,-1} \end{bmatrix}$ must be parallel to the eigenvector of $T_L$ whose corresponding eigenvalue has absolute value less than 1 in order for $\mathbf{u}_j$ to vanish as $j \to -\infty$.

To find the eigenvalues of $T_R$ and $T_L$ with absolute value less than 1, we have
\begin{align*}
\left(2-\omega^2\right)^2 - (1-\gamma_R)^2 - (1+\gamma_R)^2 &\leq \left(2\min\{\left|\gamma_R\right|,\left|\gamma_L\right|\}\right)^2- (1-\gamma_R)^2 - (1+\gamma_R)^2\\
&\leq 4\gamma_R^2 - 1+2\gamma_R-\gamma_R^2 - 1 - 2\gamma_R - \gamma_R^2\\
&= 2(\gamma_R^2 -1)\leq 0.
\end{align*}
Thus $\lambda_{-}\left(\gamma_R\right)< \lambda_{+}\left(\gamma_L\right)<0$ and the eigenvalue of $T_R$ with absolute value less than $1$ is
\begin{align*}
\lambda_R := \lambda_{+,R} = \frac{1}{2(1-\gamma_R^2)}\left(\left(2-\omega^2\right)^2 - (1-\gamma_R)^2 - (1+\gamma_R)^2 + \omega\sqrt{\left(\omega^2-4\right)\left[\left(2-\omega^2\right)^2-4\gamma_R^2\right]}\right).
\end{align*}
The corresponding eigenvector $\mathbf{e}_R$ is
\begin{align}\label{one_d_evector_right}
\mathbf{e}_R := \mathbf{e}_{+,R} = \begin{bmatrix}
(2-\omega^2)\\
1+\gamma_R+\left(1-\gamma_R\right)\lambda_R
\end{bmatrix}.
\end{align}
By similar calculations, we obtain
\begin{align*}
\lambda_L:= \frac{1}{2(1-\gamma_L^2)}\left(\left(2-\omega^2\right)^2 - (1-\gamma_L)^2 - (1+\gamma_L)^2 + \omega\sqrt{\left(\omega^2-4\right)\left[\left(2-\omega^2\right)^2-4\gamma_L^2\right]}\right)
\end{align*}
and the corresponding eigenvector
\begin{align}\label{one_d_evector_left}
\mathbf{e}_L = \begin{bmatrix}
1+\gamma_L + \left(1-\gamma_L\right)\lambda_L\\
2-\omega^2
\end{bmatrix}.
\end{align}
Therefore, we obtain that $\mathbf{u}_1 = \begin{bmatrix} u_{a,1}\\ u_{b,1} \end{bmatrix}$ must be parallel to $\mathbf{e}_R$ and $\mathbf{u}_{-1}=\begin{bmatrix} u_{a,-1}\\ u_{b,-1} \end{bmatrix}$ must be parallel to $\mathbf{e}_L$ to vanish as $j \to \pm\infty$.

\subsubsection{Proof of Theorem~\ref{one_d_main_thm}}

By substituting (\ref{1dsoln}) into the equations (\ref{1dmiddle_eqn}), we get
\begin{equation}\label{1d_middle_eqn}
\begin{aligned}
\left(2-\omega^2 + \gamma_R - \gamma_L\right) u_{a,1} - \left(1+\gamma_R\right)u_{b,1} - \left(1-\gamma_L\right)u_{b,-1} &=0,\\
\left(2-\omega^2\right)u_{b,-1} - \left(1-\gamma_L\right)u_{a,1} - \left(1+\gamma_L\right)u_{a,-1} &= 0.
\end{aligned}
\end{equation}
Since $\begin{bmatrix} u_{a,1}\\ u_{b,1} \end{bmatrix}$ and $\begin{bmatrix} u_{a,-1}\\ u_{b,-1} \end{bmatrix}$ are parallel to $\mathbf{e}_R$ and $\mathbf{e}_L$ respectively, there holds
\begin{align*}
\begin{bmatrix}
u_{a,1}\\
u_{b,1}
\end{bmatrix} = c_1 \mathbf{e}_R \,\, \text{and} \,\, 
\begin{bmatrix}
u_{a,-1}\\
u_{b,-1}
\end{bmatrix} = c_2 \mathbf{e}_L
\end{align*}
for some constants $c_1$ and $c_2$. Then by (\ref{1d_middle_eqn}), we obtain
\begin{equation}
\begin{aligned}\label{1dmiddle_last}
\left(2-\omega^2 + \gamma_R - \gamma_L\right) c_1 \left(2-\omega^2\right) - \left(1+\gamma_R\right)c_1 \left(1+\gamma_R + \left(1-\gamma_R\right)\lambda_R\right) - \left(1-\gamma_L\right)c_2 \left(2-\omega^2\right) &=0,\\
\left(2-\omega^2\right)c_2\left(2-\omega^2\right) - \left(1-\gamma_L\right)c_1 \left(2-\omega^2\right) - \left(1+\gamma_L\right)c_2 \left(1+\gamma_L+\left(1-\gamma_L\right)\lambda_L\right) &= 0.
\end{aligned}
\end{equation}
(\ref{1dmiddle_last}) can be simplified as
\begin{equation}\label{c_1c_2_eqn1}
\begin{aligned}
\left[\Omega^2 + \left(\gamma_R - \gamma_L\right)\Omega - \left(1+\gamma_R\right)^2 - \left(1-\gamma_R^2\right)\lambda_R\right]c_1 &= \left(1-\gamma_L\right) \Omega c_2,\\
\left[\Omega^2 - \left(1+\gamma_L\right)^2 - \left(1-\gamma_L^2\right)\lambda_L\right]c_2 &= \left(1-\gamma_L\right)\Omega c_1,
\end{aligned}
\end{equation}
where $\Omega = 2-\omega^2$. Multiplying the second equation in (\ref{c_1c_2_eqn1}) by $\left(1-\gamma_L\right)\Omega$ gives 
\begin{align}\label{c_1c_2_eqn2}
\left[\Omega^2 - \left(1+\gamma_L\right)^2 - \left(1-\gamma_L^2\right)\lambda_L\right]\left(1-\gamma_L\right)\Omega c_2 &= \left(1-\gamma_L\right)^2\Omega^2 c_1.
\end{align}
By the first equation in (\ref{c_1c_2_eqn1}) and (\ref{c_1c_2_eqn2}), we have
\begin{align*}
\left[\Omega^2 - \left(1+\gamma_L\right)^2 - \left(1-\gamma_L^2\right)\lambda_L\right]\left[\Omega^2 + \left(\gamma_R-\gamma_L\right)\Omega - \left(1+\gamma_R\right)^2 - \left(1-\gamma_R^2\right)\lambda_R\right] = \left(1-\gamma_L\right)^2\Omega^2,
\end{align*}
which can be simplified as
\begin{align}\label{final_eqn_1d}
\left[\Omega^2 - 4\gamma_L - \sqrt{\left(4-\Omega^2\right)\left(4\gamma_L^2-\Omega^2\right)}\right]\left[\Omega^2 + 2\left(\gamma_R-\gamma_L\right)\Omega - 4\gamma_R - \sqrt{\left(4-\Omega^2\right)\left(4\gamma_R^2-\Omega^2\right)}\right] = 4\left(1-\gamma_L\right)^2\Omega^2.
\end{align}
Observe that, for $\Omega=0$, we have
\begin{align*}
\left[0^2 - 4\gamma_L - \sqrt{\left(4-0^2\right)\left(4\gamma_L^2-0^2\right)}\right]\left[0^2 + 2\left(\gamma_R-\gamma_L\right)0 - 4\gamma_R - \sqrt{\left(4-0^2\right)\left(4\gamma_R^2-0^2\right)}\right] =  4\left(1-\gamma_L\right)^20^2
\end{align*}
which is equivalent to
\begin{align}\label{kok2}
\left(\gamma_L +|\gamma_L|\right)\left(\gamma_R+|\gamma_R|\right) = 0.
\end{align}
Since $\gamma_L$ and $\gamma_R$ have different signs, either $\gamma_L +|\gamma_L| = 0$ or $\gamma_R+|\gamma_R| = 0$. Therefore, (\ref{kok2}) holds and $\Omega=0$ is a solution of (\ref{final_eqn_1d}). Note that $\Omega = 0$ gives $\omega = \sqrt{2}\in I$. This proves the theorem.

\section{Two-dimensional Honeycomb Topological Mechanical System}\label{2d}

\begin{figure}[h]
\centering     
\subfigure[The periodic mechanical system.]{\label{gen_stru}\includegraphics[width=90mm]{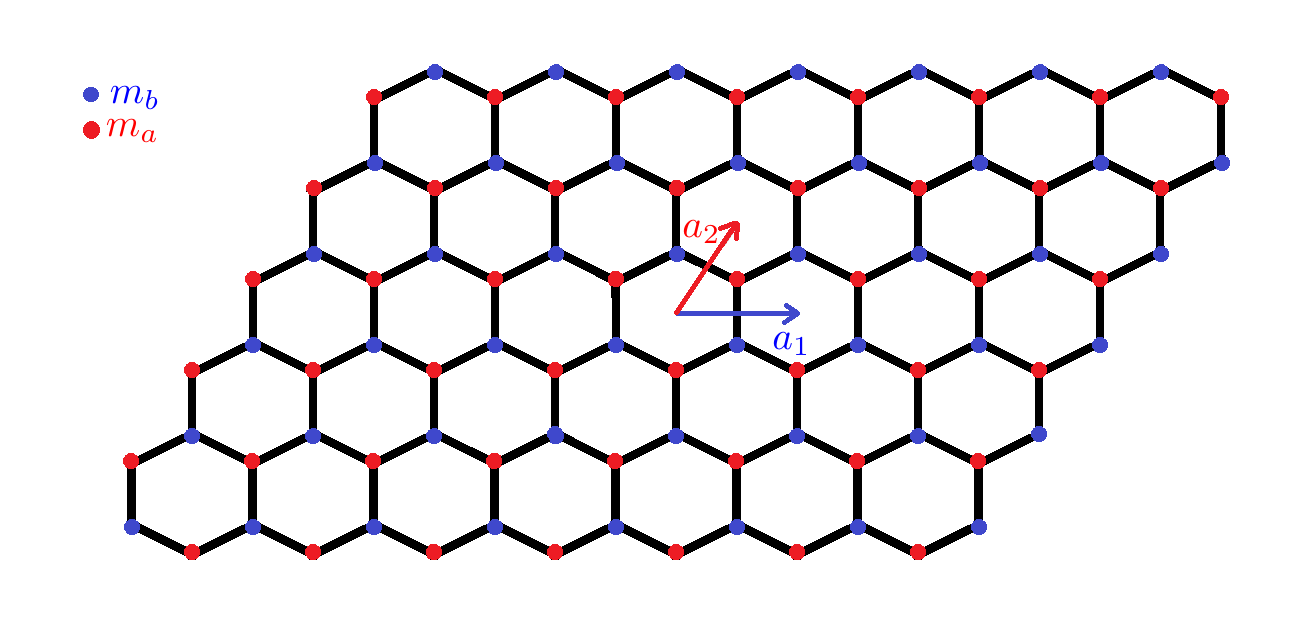}}
\subfigure[Periodic cell $Y_{p,q} = \{(p+t_1)\mathbf{a_1} + (q+t_2)\mathbf{a_2} : 0 \leq t_1,t_2 \leq 1\}$.]{\label{cell_numb}\includegraphics[width=70mm]{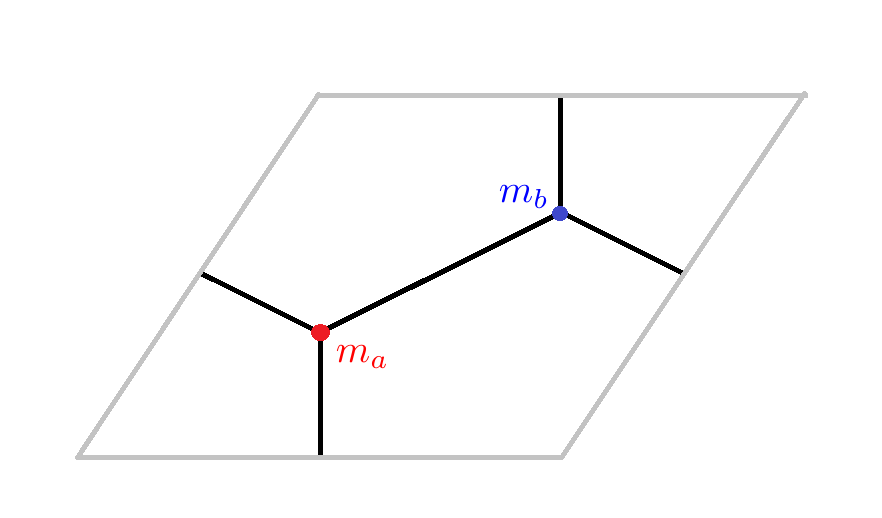}}
\caption{The periodic mechanical system over the honeycomb lattice.}
\label{gencase2d}
\end{figure}

\subsection{Periodic Mechanical System}

\subsubsection{Mathematical Model}

We consider the two-dimensional  mechanical system over the honeycomb lattice as shown in Figure \ref{gen_stru}. Let $\mathbf{a}_1 = a[1, 0]$ and $\mathbf{a}_2 = a[\frac{1}{2}, \frac{\sqrt{3}}{2}]$ be the lattice vectors where $a$ is the lattice constant. Then the honeycomb lattice is given by $\Lambda := \sum\limits_{p,q \in \mathbb{Z}} Y_{p,q}$, where $Y_{p,q} = \{(p+t_1)\mathbf{a}_1 + (q+t_2)\mathbf{a}_2 : 0 \leq t_1,t_2 \leq 1\}$ as shown in Figure \ref{cell_numb}. Each periodic cell contains two masses, $m_a = m (1+\beta)$ and $m_b = m(1-\beta)$ with $-1< \beta < 1$, connected by a spring of the length $\frac{a}{\sqrt{3}}$ and the spring constant $k$. Let $\mathbf{b}_1$ and $\mathbf{b}_2$ be the reciprocal lattice vectors given by $\mathbf{b}_1 = \frac{2\pi}{a}\left[1, -\frac{1}{\sqrt{3}}\right]$ and $\mathbf{b}_2 = \frac{2\pi}{a}\left[0, \frac{2}{\sqrt{3}}\right]$, which satisfy
\begin{align*}
\mathbf{a}_i\cdot \mathbf{b}_j = 2\pi\delta_{ij} = 
\begin{cases}
0, \,\,i\neq j,\\
2\pi,\,\, i=j.
\end{cases}
\end{align*}

\begin{figure}[h]
\centering
\includegraphics[scale=0.5]{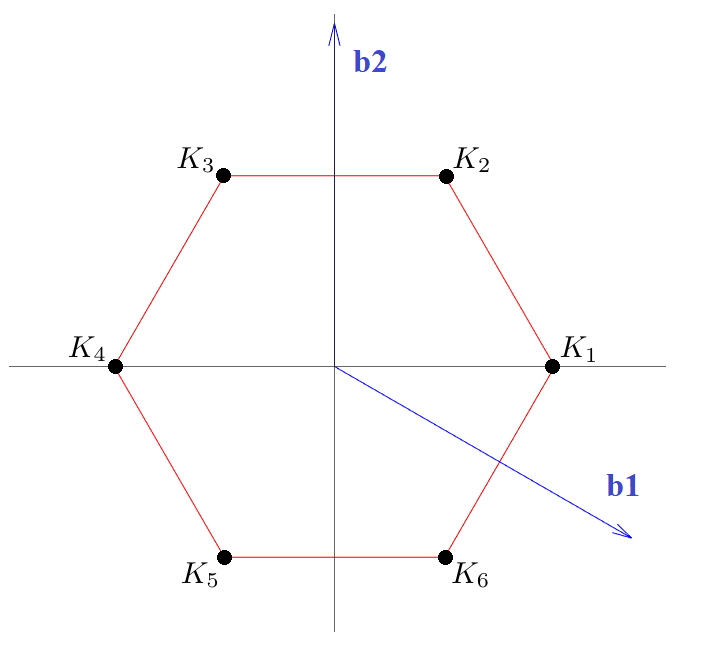}
\caption{The Brillouin zone $\mathbb{B}$ in the reciprocal lattice. The vertices of the Brillouin zone are $K_1 = \left(\frac{4\pi}{3a},0\right)$, $K_2 = \left(\frac{2\pi}{3a},\frac{2\pi}{a\sqrt{3}}\right)$, $K_3 = \left(-\frac{2\pi}{3a},\frac{2\pi}{a\sqrt{3}}\right)$, $K_4 = \left(-\frac{4\pi}{3a},0\right)$, $K_5 = \left(-\frac{2\pi}{3a},-\frac{2\pi}{a\sqrt{3}}\right)$ and 
$K_6 = \left(\frac{2\pi}{3a},-\frac{2\pi}{a\sqrt{3}}\right)$.}
\label{bzvectors}
\end{figure}

The hexagonal shape of the fundamental cell in the reciprocal lattices $ \Lambda^* :=\sum_{p,q \in \mathbb{Z}} p\mathbf{b}_1+q\mathbf{b}_2$, or the Brillouin zone $\mathbb{B}$, is shown in Figure \ref{bzvectors}. Over the periodic cell $Y_{p,q}$, the displacements $U_{p,q}^a(t)$ and $U_{p,q}^b(t)$ for the masses $a$ and $b$ satisfy
\begin{equation} \label{eqn_2d}
\begin{aligned}
m_a(U_{p,q}^a)'' + k\left(3U_{p,q}^a-U_{p,q}^b-U_{p-1,q}^b-U_{p,q-1}^b\right) &= 0,\\
m_b(U_{p,q}^b)'' + k\left(3U_{p,q}^b-U_{p,q}^a-U_{p+1,q}^a-U_{p,q+1}^a\right) &= 0.
\end{aligned}
\end{equation}
Consider the time-harmonic solution in the form of
\begin{align}\label{2dsoln}
U_{p,q}^a(t) = u_{a} \exp\left[i(\omega \tau + \boldsymbol{\kappa} \cdot \mathbf{r}_{p,q})\right]\,\,\,\text{and}\,\,\, U_{p,q}^b(t) = u_{b} \exp\left[i(\omega \tau + \boldsymbol{\kappa} \cdot \mathbf{r}_{p,q})\right],
\end{align}
where $\tau=\sqrt{k/m}t$ is nondimensional time scale, $u_a$ and $u_b$ are the amplitudes of displacements, the position vector $\mathbf{r_{p,q}} = p\mathbf{a}_1 + q\mathbf{a}_2$, the wave vector $\boldsymbol{\kappa} = \kappa_1 \mathbf{b}_1 + \kappa_2 \mathbf{b}_2$, and the wave frequency $\omega$. Then $u_a$ and $u_b$ satisfy
\begin{align*}
(1+\beta) (-\omega^2)u_a + 3u_a +(-1-e^{-i\boldsymbol{\kappa} \cdot \boldsymbol{a}_1}-e^{-i\boldsymbol{\kappa} \cdot \boldsymbol{a}_2})u_b &= 0,\\
(1-\beta)(-\omega^2)u_b + 3u_b + (-1 - e^{i\boldsymbol{\kappa} \cdot \boldsymbol{a}_1}-e^{i\boldsymbol{\kappa} \cdot \boldsymbol{a}_2})u_a &= 0,
\end{align*}
or equivalently,
\begin{align}\label{system}
M\left(\boldsymbol{\kappa}\right)\boldsymbol{v}\left(\boldsymbol{\kappa}\right) = \omega^2\boldsymbol{v}\left(\boldsymbol{\kappa}\right),
\end{align}
wherein 
\begin{align}\label{M_in_chern}
M\left(\boldsymbol{\kappa}\right):=\begin{bmatrix}
\frac{3}{1+\beta} & \frac{\overline{d\left(\kappa\right)}}{1+\beta}\\
\frac{d\left(\kappa\right)}{1-\beta} & \frac{3}{1-\beta}
\end{bmatrix} \,\,\,\text{and}\,\,\,
\boldsymbol{v}\left(\boldsymbol{\kappa}\right) := 
\begin{bmatrix}
u_a\\
u_b
\end{bmatrix}.
\end{align} 
In the above, $d\left(\boldsymbol{\kappa}\right) := -1 - e^{i\boldsymbol{\kappa} \cdot \boldsymbol{a}_1}-e^{i\boldsymbol{\kappa} \cdot \boldsymbol{a}_2} = -1 - e^{i2\pi \kappa_1} - e^{i2\pi \kappa_2}$ and $\overline{d\left(\boldsymbol{\kappa}\right)}$ is the complex conjugate of $d\left(\boldsymbol{\kappa}\right)$. The eigenvalues of the matrix $M\left(\boldsymbol{\kappa}\right)$ are
\begin{align}\label{evalue_c}
\lambda_{\pm}\left(\boldsymbol{\kappa}\right) = \frac{3}{1-\beta^2}\pm \sqrt{\frac{9\beta^2}{\left(1-\beta^2\right)^2} + \frac{|d\left(\boldsymbol{\kappa}\right)|^2}{1-\beta^2}},
\end{align}
with the corresponding eigenvectors
\begin{align}\label{evec_c}
\boldsymbol{v}_\pm\left(\boldsymbol{\kappa}\right) = \frac{1}{\chi\left(\boldsymbol{\kappa}\right)}
\begin{bmatrix}
-\frac{\overline{d\left(\boldsymbol{\kappa}\right)}}{3-\lambda_{\pm}\left(\boldsymbol{\kappa}\right)(1+\beta)}\\
1
\end{bmatrix}.
\end{align}
In the above, $\chi\left(\boldsymbol{\kappa}\right)$ is a normalization constant such that
\begin{align}\label{normalization}
\boldsymbol{v}_\pm^*\left(\boldsymbol{\kappa}\right)
\begin{bmatrix}
1+\beta & 0 \\0 & 1-\beta
\end{bmatrix}
\boldsymbol{v}_\pm\left(\boldsymbol{\kappa}\right) = 1,
\end{align}
where $\boldsymbol{v}_\pm^*\left(\boldsymbol{\kappa}\right)$ is conjugate transpose of $\boldsymbol{v}_\pm\left(\boldsymbol{\kappa}\right)$. In what follows, we use $d$, $M$, $\lambda$ and $\boldsymbol{v}$ instead of $d\left(\boldsymbol{\kappa}\right)$, $M\left(\boldsymbol{\kappa}\right)$ $\lambda\left(\boldsymbol{\kappa}\right)$ and $\boldsymbol{v}\left(\boldsymbol{\kappa}\right)$ for simplicity.

\subsubsection{Dirac Point when $\beta=0$}
\begin{figure}[h]
\centering     
\subfigure[$\beta = 0$]{\label{beta_zero}\includegraphics[width=82mm]{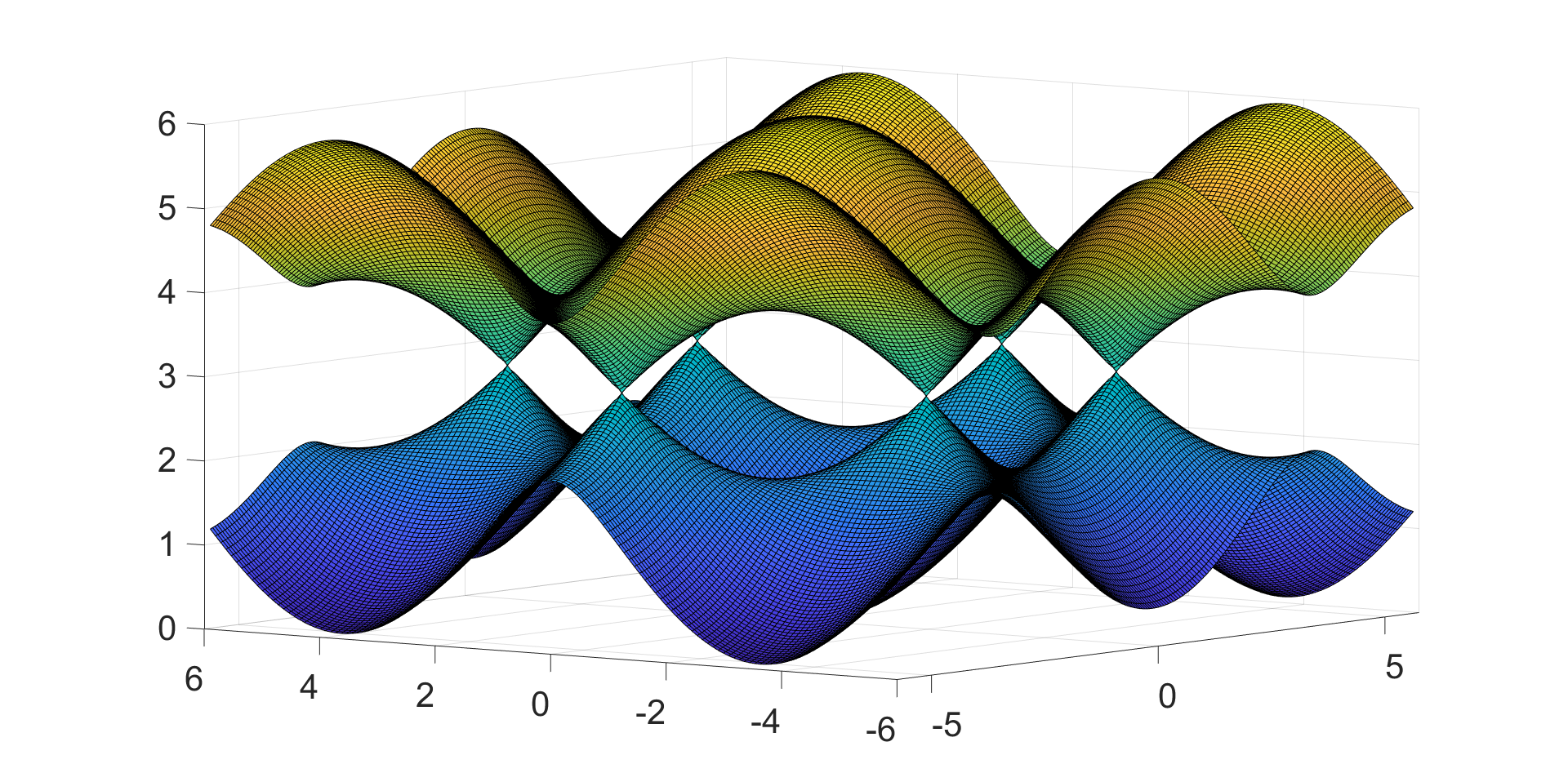}}
\subfigure[$\beta = 0.05$]{\label{beta_nonzero}\includegraphics[width=82mm]{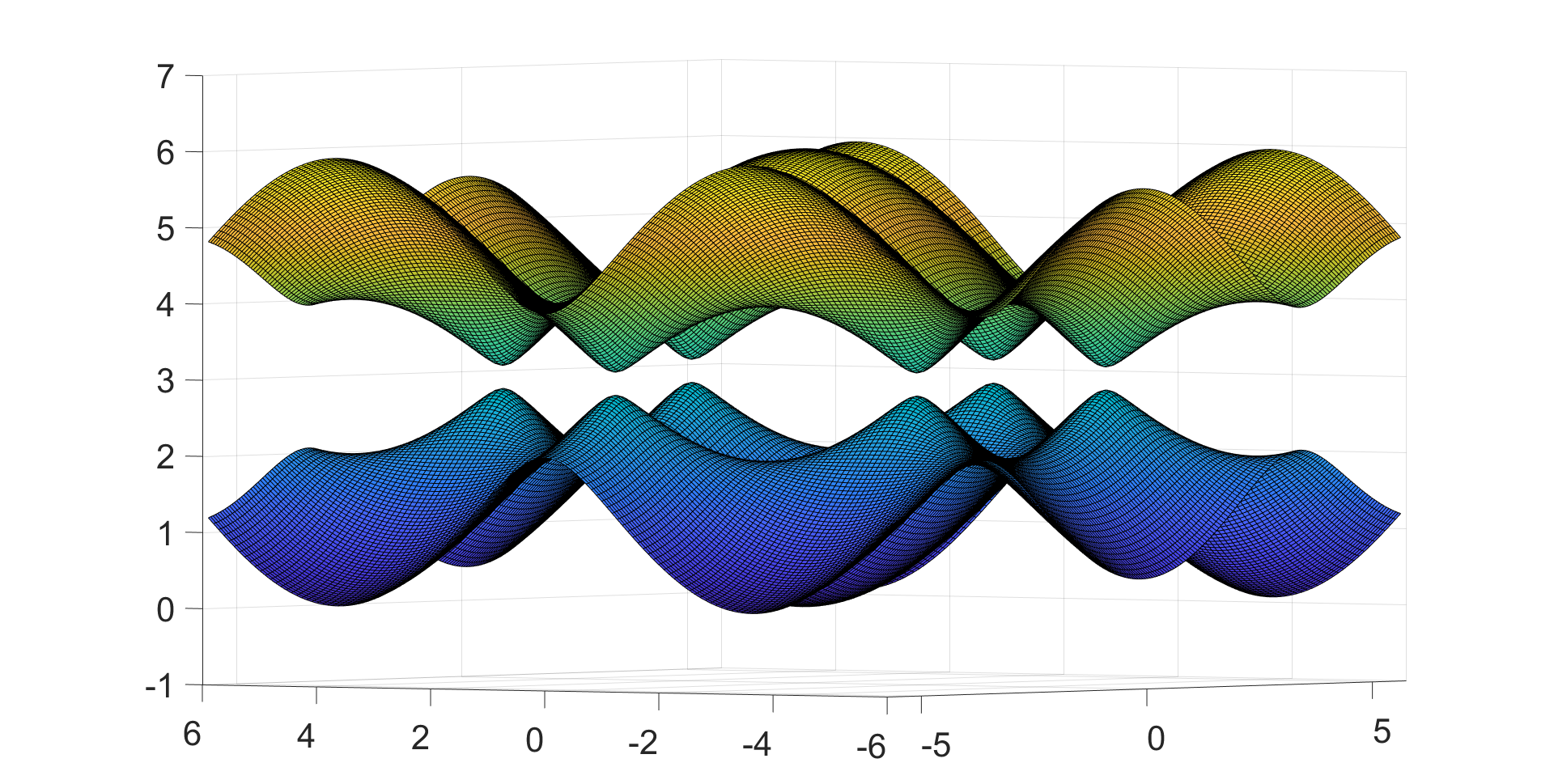}}
\caption{The band structure of the periodic system. (a) $\beta=0$: The upper and lower bands touch at the vertices of Brillouin zone and form the Dirac points. (b) $\beta \neq 0$: A gap opens between the two bands.}
\label{fig:globfig}
\end{figure}

We first study the band structure when the two masses $m_a=m_b$, namely when $\beta=0$. In particular, we show that Dirac point exists at the vertices of the Brillouin zone. A pair $\left(\boldsymbol{\kappa}^*,\lambda^*\right) \in \mathbb{B} \times \mathbb{R}$ is called a Dirac point (cf. \cite{Lin_2022},\cite{lee2019elliptic},\cite{fefferman2018honeycomb}, \cite{fefferman2014wave}) if
\begin{enumerate}
\item $\lambda_+\left(\boldsymbol{\kappa}^*\right)=\lambda_-\left(\boldsymbol{\kappa}^*\right) = \lambda^*$. In addition, there exist constants $\alpha>0$ and $\gamma>0$ such that the expansions
\begin{align*}
\lambda_+\left(\boldsymbol{\kappa}\right) &= \lambda^* + \alpha|\boldsymbol{\kappa}-\boldsymbol{\kappa}^*| + O\left(|\boldsymbol{\kappa}|^2\right),\\
\lambda_-\left(\boldsymbol{\kappa}\right) &= \lambda^* - \alpha|\boldsymbol{\kappa}-\boldsymbol{\kappa}^*| + O\left(|\boldsymbol{\kappa}|^2\right),
\end{align*}
hold for $|\boldsymbol{\kappa} - \boldsymbol{\kappa}^*|<\gamma$.
\item The eigenvalue $\lambda_{\pm}\left(\boldsymbol{\kappa}\right)$ in (\ref{evalue_c}) has multiplicity two when $\boldsymbol{\kappa} = \boldsymbol{\kappa}^*$.
\end{enumerate}

\begin{remark}\label{valley_d}
Observe that if $\boldsymbol{\kappa}_i = K_i$, then $\boldsymbol{\kappa}_1 = \boldsymbol{\kappa}_3+\boldsymbol{b}_1$, $\boldsymbol{\kappa}_3 = \boldsymbol{\kappa}_5+\boldsymbol{b}_2$, $\boldsymbol{\kappa}_2 = \boldsymbol{\kappa}_6+\boldsymbol{b}_2$ and $\boldsymbol{\kappa}_6 = \boldsymbol{\kappa}_4+\boldsymbol{b}_1$. Therefore,
\begin{align}\label{equality_of_d}
d\left(\boldsymbol{\kappa}_1\right) = d\left(\boldsymbol{\kappa}_3\right) = d\left(\boldsymbol{\kappa}_5\right) \,\,\,\text{and}\,\,\,d\left(\boldsymbol{\kappa}_2\right) = d\left(\boldsymbol{\kappa}_4\right) = d\left(\boldsymbol{\kappa}_6\right),
\end{align}
and it is sufficient to study the eigenvalues for $\boldsymbol{\kappa}=K_1$ and $\boldsymbol{\kappa}=K_4$.
\end{remark}
When $\beta=0$, from (\ref{evalue_c}), the eigenvalues of $M$ in (\ref{system}) are
\begin{align*}
\lambda\left(\boldsymbol{\kappa}\right)= 3 \pm \left|d\left(\boldsymbol{\kappa}\right)\right|.
\end{align*}
Observe that
\begin{align*}
d\left(K_1\right) = -1-\exp\left(ia\frac{4\pi}{3a}\right) - \exp\left(ia\frac{\frac{4\pi}{3a} + 0\sqrt{3}}{2}\right)  = 0.
\end{align*}
We obtain $\lambda_+\left(K_1\right) = \lambda_-\left(K_1\right) = 3$ (Figure \ref{beta_zero}). In addition, the derivative of $\lambda_+\left(\boldsymbol{\kappa}\right)$ at $\boldsymbol{\kappa}=K_1$ along the direction $\mathbf{w} = (w_1,w_2)$ is
\begin{align*}
D_{\boldsymbol{w}}\lambda_+\left(K_1\right) &=\lim_{h \to 0^+}\frac{1}{h} \left[\lambda_+\left(\frac{4\pi}{3a} + w_1h, 0 +w_2h\right) - \lambda_+\left(\frac{4\pi}{3a}, 0\right)\right]\\
&=\lim_{h \to 0^+}\sqrt{\frac{1}{h^2} \left[4\cos^2\left[\frac{2\pi}{3}+\frac{w_1a}{2}h\right] + 4 \cos\left[\frac{2\pi}{3}+\frac{w_2a}{2}h\right]+ 1\right]} = \frac{a\sqrt{3}}{2}.
\end{align*}
Similarly,
\begin{align*}
D_{\boldsymbol{w}}\lambda_-\left(K_1\right) = - \frac{a\sqrt{3}}{2}.
\end{align*}
Therefore, near $K_1$, there holds
\begin{align*}
\lambda_{\pm}\left(\boldsymbol{\kappa}\right) &= 3 \pm \frac{a\sqrt{3}}{2} \left|\boldsymbol{\kappa} - K_1\right| + O\left(\left|\boldsymbol{\kappa} - K_1\right|^2\right).
\end{align*}
Following similar calculations, it can be shown that, for $\boldsymbol{\kappa}$ near $K_4$,
\begin{align*}
\lambda_{\pm}\left(\boldsymbol{\kappa}\right) &= 3 \pm \frac{a\sqrt{3}}{2} \left|\boldsymbol{\kappa} - K_4\right| + O\left(\left|\boldsymbol{\kappa} - K_4\right|^2\right).
\end{align*}
Note that have $M\left(\boldsymbol{\kappa}\right) = \begin{bmatrix}
3 & 0 \\
0 & 3
\end{bmatrix}$ for $\boldsymbol{\kappa} = K_1, K_4$. Thus, the multiplicity of $\lambda^*$ is 2. Therefore, $(K_i,3)$ is a Dirac point for $i=1,4$.
\begin{remark}\label{bandgap_remark}
When $\beta \neq 0$, we have $\lambda_-\left(\boldsymbol{\kappa}\right) < \lambda_+\left(\boldsymbol{\kappa}\right)$ for $\boldsymbol{\kappa} \in \mathbb{B}$ and there is a gap between the upper and lower bands $\lambda_\pm\left(\boldsymbol{\kappa}\right)$ in (\ref{evec_c}) which is called band gap.
\end{remark}

\subsubsection{Valley Chern Number}

The Berry phase is a phase angle that describes the global phase change of a complex vector over a closed loop $\nu$ in its parameter space. The Berry phase associated with the band $\lambda$ of the system in (\ref{system}) is defined as a line integral around a closed loop $\nu$ in the Brillouin zone (cf. \cite{vanderbilt});
\begin{align}\label{berry_phase}
\phi = \oint_{\nu} \boldsymbol{B}(\boldsymbol{\kappa})  d\boldsymbol{\kappa}.
\end{align}
In the above, $\boldsymbol{B}(\boldsymbol{\kappa}) = (A_{\kappa_x}\left(\boldsymbol{\kappa}\right),A_{\kappa_y}\left(\boldsymbol{\kappa}\right))$ is the Berry connection, wherein 
\begin{align*}
A_j(\boldsymbol{\kappa}) := \langle\boldsymbol{v}(\boldsymbol{\kappa}),i\partial_j \boldsymbol{v}(\boldsymbol{\kappa})\rangle,\,\,\, j = \kappa_x, \kappa_y,
\end{align*}
and $\boldsymbol{v}$ is the eigenvector of $M$ associated with the eigenvalue $\lambda$ as defined in (\ref{evec_c}). By the Stokes' theorem,
\begin{align}\label{berry_curv}
\phi = \int_{D} \Omega(\boldsymbol{\kappa}) dS,
\end{align}
where $D$ is the region enclosed by $\nu$ and $\Omega(\boldsymbol{\kappa})$ is Berry curvature given by $\Omega\left(\boldsymbol{\kappa}\right) = \partial_{\kappa_x} A_{\kappa_y}\left(\boldsymbol{\kappa}\right) - \partial_{\kappa_y}A_{\kappa_x}\left(\boldsymbol{\kappa}\right)$. The valley Chern number for a Bloch wave vector $\boldsymbol{\kappa} = K$ is defined as Berry phase calculated over a closed loop $\nu$ containing $K$ scaled by $2\pi$ (cf. \cite{xiao2007valley}), i.e.,
\begin{align}\label{valley_defn}
C_{K,\nu} = \frac{1}{2\pi}\left(\oint_{\nu} B(\boldsymbol{\kappa}) d\boldsymbol{\kappa} \pmod{2\pi} \right) = \frac{1}{2\pi}\int_{D} \Omega(\boldsymbol{\kappa}) dS.
\end{align}
\begin{remark}
For the eigenvector $\boldsymbol{v}$ of the system in (\ref{system}), a gauge transformation $\tilde{\boldsymbol{v}}\left(\boldsymbol{\kappa}\right) = e^{-i\varphi\left(\boldsymbol{\kappa}\right)}\boldsymbol{v}\left(\boldsymbol{\kappa}\right)$ for a differentiable function $\varphi\left(\boldsymbol{\kappa}\right)$ gives that $\tilde{\boldsymbol{B}}\left(\boldsymbol{\kappa}\right) = \boldsymbol{B}\left(\boldsymbol{\kappa}\right) + \nabla \varphi\left(\boldsymbol{\kappa}\right)$. Hence the Berry phase $\tilde{\phi} = \phi+2\pi m$, for some $m \in \mathbb{Z}$, but $\tilde{\Omega}\left(\boldsymbol{\kappa}\right) = \Omega\left(\boldsymbol{\kappa}\right)$ since $\nabla \times\nabla \varphi = 0$. As such (\ref{valley_defn}) is defined with modulo $2\pi$.
\end{remark}
Let $\nu$ be a circle centered at $K$ with radius $0 < r \ll 1$, we define the discrete valley Chern number as, for $N \in \mathbb{Z}^+$,
\begin{align}\label{discrete_chern_formula}
C^N_{K,\nu} = -\frac{1}{2\pi}\left[Im\left( \sum_{j=1}^{N} \log \langle\boldsymbol{v}\left(K+\boldsymbol{\kappa}_j\right), \boldsymbol{v}\left(K+\boldsymbol{\kappa}_{j+1}\right)\rangle \right)\pmod{2\pi}\right],
\end{align}
where $\boldsymbol{\kappa}_j = r\left(\cos\left(\theta_j\right),\sin\left(\theta_j\right)\right)$, $\theta_j = -\pi+(j-1)\frac{2\pi}{N}$. It is clear that
\begin{align*}
C_{K,\nu} = \lim\limits_{N \to \infty} C_{K,\nu}^N.
\end{align*}
In what follows, we use $C_K$ instead of $C_{K,\nu}$ for simplicity.

\begin{lemma}\label{chern_lemma}\leavevmode
For $N = 2n \in \mathbb{Z}^+$,
\begin{enumerate}
\item The Berry phase $\phi$ over the Brillouin zone is zero.
\item The valley Chern numbers for the vertices of Brillouin zone satisfy
\begin{align*}
C^N_{K_1} = C^N_{K_3} = C^N_{K_5}\,\,\, \text{and}\,\,\, C^N_{K_2}=C^N_{K_4}=C^N_{K_6}.
\end{align*}
In addition, the valley Chern numbers satisfy $C^N_{K_1} = -C^N_{K_4}$.
\item Let $\boldsymbol{v}_\pm$ be the eigenvectors of M defined in (\ref{evec_c}). If $\boldsymbol{v} = \boldsymbol{v}_+$, $\beta$ and $C_{K_1}^N$ attain opposite signs, and if $\boldsymbol{v} = \boldsymbol{v}_-$, $\beta$ and $C_{K_1}^N$ attain the same signs, where $\boldsymbol{v}_\pm$ is the vector given in (\ref{evec_c}).
\end{enumerate}
\end{lemma}

\begin{proof}
$(i)$ For $N = 2n$ with $n \in \mathbb{Z}^+$, let $\left\{\boldsymbol{\hat{\kappa}}_j\right\}_{j=1}^N$ be equally spaced points on the boundary of Brillouin zone such that $\{K_1,K_2,...K_6\} \subset \left\{\boldsymbol{\hat{\kappa}}_j\right\}_{j=1}^N$. Then, for $j=1,2,...,n+1$, we have 
\begin{align*}
\hat{\boldsymbol{\kappa}}_j = -\hat{\boldsymbol{\kappa}}_{n+j},
\end{align*}
which implies
\begin{align*}
d\left(\hat{\boldsymbol{\kappa}}_{j+n}\right) = \overline{d\left(\hat{\boldsymbol{\kappa}}_{j}\right)}
\end{align*}
and thus
\begin{align*}
\boldsymbol{v}\left(\hat{\boldsymbol{\kappa}}_{j\pm n}\right) = \overline{\boldsymbol{v}\left(\hat{\boldsymbol{\kappa}}_{j}\right)}
\end{align*}
for $1\leq j \leq n+1$ in $\boldsymbol{v}\left(\hat{\boldsymbol{\kappa}}_{j+ n}\right)$ and for $n+2\leq j \leq N$ in $\boldsymbol{v}\left(\hat{\boldsymbol{\kappa}}_{j- n}\right)$. Thus we have
\begin{align*}
C^N &= -Im\left(\sum_{j=1}^N \log\langle\mathbf{v}\left(\hat{\boldsymbol{\kappa}}_{j}\right)|\mathbf{v}\left(\hat{\boldsymbol{\kappa}}_{j+1}\right)\rangle\right)\\
&= -Im\left(\sum_{j=1}^{n} \log\langle\mathbf{v}\left(\hat{\boldsymbol{\kappa}}_{j}\right)|\mathbf{v}\left(\hat{\boldsymbol{\kappa}}_{j+1}\right)\rangle + \sum_{j=n+1}^N \log\langle\mathbf{v}\left(\hat{\boldsymbol{\kappa}}_{j}\right)|\mathbf{v}\left(\hat{\boldsymbol{\kappa}}_{j+1}\right)\rangle\right)\\
&= -Im\left(\sum_{j=1}^{n} \log\langle\mathbf{v}\left(\hat{\boldsymbol{\kappa}}_{n+j}\right)^*|\mathbf{v}\left(\hat{\boldsymbol{\kappa}}_{j}\right)_{N/2+j+1}^*\rangle + \sum_{j=n+1}^N \log\langle\mathbf{v}\left(\hat{\boldsymbol{\kappa}}_{j}\right)|\mathbf{v}\left(\hat{\boldsymbol{\kappa}}_{j+1}\right)\rangle\right)\\
&= -Im\left(\sum_{n=n+1}^{N} \log\langle\mathbf{v}\left(\hat{\boldsymbol{\kappa}}_{n}\right)^*|\mathbf{v}\left(\hat{\boldsymbol{\kappa}}_{n+1}\right)^*\rangle + \sum_{j=n+1}^N \log\langle\mathbf{v}\left(\hat{\boldsymbol{\kappa}}_{j}\right)|\mathbf{v}\left(\hat{\boldsymbol{\kappa}}_{j+1}\right)\rangle\right)\\
&= -Im\left(\sum_{n=n+1}^{N}\overline{ \log\langle\mathbf{v}\left(\hat{\boldsymbol{\kappa}}_{n}\right)|\mathbf{v}\left(\hat{\boldsymbol{\kappa}}_{n+1}\right)\rangle} + \sum_{j=n+1}^N \log\langle\mathbf{v}\left(\hat{\boldsymbol{\kappa}}_{j}\right)|\mathbf{v}\left(\hat{\boldsymbol{\kappa}}_{j+1}\right)\rangle\right)\\
&=0.
\end{align*}
$(ii)$ By Remark \ref{valley_d}, $d\left(K_1\right) = d\left(K_3\right) =d\left(K_5\right)$. Therefore, $\mathbf{v}\left(K_1\right) = \mathbf{v}\left(K_3\right) = \mathbf{v}\left(K_5\right)$ which implies $C^N_{K_1} = C^N_{K_3} = C^N_{K_5}$. Similarly, $C^N_{K_2} = C^N_{K_4} = C^N_{K_6}$. Let $\left\{\boldsymbol{\kappa}_j\right\}_{j=1}^N$ be
the equally spaced points on the circle as given in the definition of the discrete valley Chern number, wherein $N = 2n$ for some $n \in \mathbb{Z}^+$. Note that $arg\left(\boldsymbol{\kappa}_{j+n}\right) = arg\left(\boldsymbol{\kappa}_{j}\right)+\pi$ for $j = 1,2,...,n$, thus
\begin{align}\label{circ_valley}
K_4+\boldsymbol{\kappa}_{j+n} = -\left(K_1 + \boldsymbol{\kappa}_{j}\right)\,\,\,\text{for}\,\,\,1\leq j \leq n\,\,\,\text{and}\,\,\,K_4 + \boldsymbol{\kappa}_{j-n} = -\left(K_1+\boldsymbol{\kappa}_{j}\right)\,\,\,\text{for}\,\,\,n+1\leq j \leq N.
\end{align}
Then,
\begin{align*}
d\left(K_4+\boldsymbol{\kappa}_{j+n}\right) = \overline{d\left(K_1+\boldsymbol{\kappa}_{j}\right)}\,\,\,\text{and}\,\,\,d\left(K_4+\boldsymbol{\kappa}_{j-n}\right) = \overline{d\left(K_1+\boldsymbol{\kappa}_{j}\right)},
\end{align*}
which implies 
\begin{align*}
\boldsymbol{v}\left(K_4+\boldsymbol{\kappa}_{j\pm n}\right) = \overline{\boldsymbol{v}\left(K_1+\boldsymbol{\kappa}_{j}\right)}\,\,\,\text{for}\,\,\,1\leq j \leq n,\,\,\,\text{and}\,\,\,n+1\leq j \leq N.
\end{align*}
Then 
\begin{align*}
C_4^N &= -Im \left(\sum_{j=1}^{2n} \log\langle\boldsymbol{v}\left(K_4 +\boldsymbol{\kappa}_j\right)| \boldsymbol{v}\left(K_4 +\boldsymbol{\kappa}_{j+1}\right)\rangle\right)\\
&=-Im \left(\sum_{j=1}^{n} \log\langle\boldsymbol{v}\left(K_4 +\boldsymbol{\kappa}_j\right)| \boldsymbol{v}\left(K_4 +\boldsymbol{\kappa}_{j+1}\right)\rangle + \sum_{j=n+1}^{2n}\log\langle\boldsymbol{v}\left(K_4 +\boldsymbol{\kappa}_j\right)| \boldsymbol{v}\left(K_4 +\boldsymbol{\kappa}_{j+1}\right)\rangle \right)\\
&=-Im \left(\sum_{m=n+1}^{2n} \log\langle\boldsymbol{v}\left(K_4 +\boldsymbol{\kappa}_{m-n}\right)| \boldsymbol{v}\left(K_4 +\boldsymbol{\kappa}_{m-n+1}\right)\rangle + \sum_{m=1}^{n}\log\langle\boldsymbol{v}\left(K_4 +\boldsymbol{\kappa}_{m+n}\right)| \boldsymbol{v}\left(K_4 +\boldsymbol{\kappa}_{m+n+1}\right)\rangle \right)\\
&=-Im \left(\sum_{m=n+1}^{2n} \log\langle\boldsymbol{v}^*\left(K_1 +\boldsymbol{\kappa}_{m}\right)| \boldsymbol{v}^*\left(K_1 +\boldsymbol{\kappa}_{m+1}\right)\rangle + \sum_{m=1}^{n}\log\langle\boldsymbol{v}^*\left(K_1 +\boldsymbol{\kappa}_{m}\right)| \boldsymbol{v}^*\left(K_1 +\boldsymbol{\kappa}_{m+1}\right)\rangle \right)\\
&=-Im \left(\left[\sum_{m=1}^{2n} \log\langle\boldsymbol{v}\left(K_1 +\boldsymbol{\kappa}_{m}\right)| \boldsymbol{v}\left(K_1 +\boldsymbol{\kappa}_{m+1}\right)\rangle\right]^*\right) = -C_1^N.
\end{align*}\\
$(iii)$ Let $\left\{\boldsymbol{\kappa}_j\right\}_{j=1}^N$ be as in $(ii)$. Then we have
\begin{align*}
\langle\boldsymbol{v}\left(\boldsymbol{\kappa}_j\right)|\boldsymbol{v}\left(\boldsymbol{\kappa}_{j+1}\right)\rangle = \left(1-\beta\right) |d\left(\boldsymbol{\kappa}_j\right)||d\left(\boldsymbol{\kappa}_{j+1}\right)|f_j\left(\beta\right) \left[r_j\left(\beta\right) e^{i\left(\xi_j - \xi_{j+1}\right)} +1 \right],
\end{align*}
where
\begin{align*}
f_j\left(\beta\right) &= \frac{\left(3\beta+\sqrt{9\beta^2+\left(1-\beta^2\right)|d\left(\boldsymbol{\kappa}_j\right)|^2}\right)\left(3\beta+\sqrt{9\beta^2+\left(1-\beta^2\right)|d\left(\boldsymbol{\kappa}_{j+1}\right)|^2}\right)}{|d\left(\boldsymbol{\kappa}_j\right)||d\left(\boldsymbol{\kappa}_{j+1}\right)|},\\
r_j\left(\beta\right) &= \frac{1-\beta^2}{f_j\left(\beta\right)},\,\,\, \xi_j = arg\left(d\left(\boldsymbol{\kappa}_j\right)\right).
\end{align*}
For $\beta>0$, $f_j\left(\beta\right) \in [1,\frac{36}{|d\left(\boldsymbol{\kappa}_j\right)||d\left(\boldsymbol{\kappa}_{j+1}\right)|}]$ is an increasing function of $\beta$ and
\begin{align*}
\sqrt{1-\beta^2}\left|d\left(\boldsymbol{\kappa}_j\right)\right| > 3\beta + \sqrt{9\beta^2 + \left(1-\beta^2\right)\left|d\left(\boldsymbol{\kappa}_j\right)\right|^2},
\end{align*}
thus $r_j\left(\beta\right) \in [0,1]$. Consequently,
\begin{align}\label{dot_chern}
Im\left( \sum_{j=1}^{N} \log \langle\boldsymbol{v}\left(\boldsymbol{\kappa}_j\right), \boldsymbol{v}\left(\boldsymbol{\kappa}_{j+1}\right)\rangle\right) = Im\left( \sum_{j=1}^{N} \log \left(r_j\left(\beta\right)e^{i\left(\xi_j-\xi_{j+1}\right)} + 1\right) \right) = \sum_{j=1}^{N} \varepsilon_j,
\end{align}
where $\varepsilon_j \in \left(0,\frac{\xi_j-\xi_{j+1}}{2}\right)$ and we have used Lemma~\ref{arg}. Denoting $d\left(\boldsymbol{\kappa}\left(\theta\right)\right)$ as $d\left(\theta\right)$, it follows that
\begin{align*}
d\left(-\frac{\pi}{3}-\tilde\theta\right) = \overline{d\left(-\frac{\pi}{3}+\tilde\theta\right)}\,\,\,\text{and}\,\,\,d\left(\frac{2\pi}{3}-\tilde\theta\right) = \overline{d\left(\frac{2\pi}{3}+\tilde\theta\right)},
\end{align*}
for $\tilde\theta \in \left[0,\frac{2\pi}{3}\right]$ and $\tilde\theta \in \left[0,\frac{\pi}{3}\right]$ respectively. In addition, $d(-\pi/3)d(2\pi/3) < 0$. Therefore, as $\theta \in [-\pi,\pi]$, $d\left(\theta\right)$ surrounds the origin on the complex plane and $arg\left(d\left(-\pi\right)\right) = arg\left(d\left(\pi\right)\right)+2\pi$. Then, by (\ref{dot_chern}),
\begin{align*}
Im\left( \sum_{j=1}^{N} \log \langle\boldsymbol{v}\left(\boldsymbol{\kappa}_j\right), \boldsymbol{v}\left(\boldsymbol{\kappa}_{j+1}\right)\rangle\right) < \sum_{j=1}^{N} \frac{\xi_j - \xi_{j+1}}{2} = \frac{\xi_1-\xi_{N+1}}{2} = \pi.
\end{align*}
Therefore,
\begin{align}\label{sum_pos_beta}
Im\left( \sum_{j=1}^{N} \log \langle\boldsymbol{v}\left(\boldsymbol{\kappa}_j\right), \boldsymbol{v}\left(\boldsymbol{\kappa}_{j+1}\right)\rangle\right) \in \left(0,\pi\right).
\end{align}
For $\beta<0$, there holds $r_j\left(\beta\right) > 1$. By similar calculations, we obtain
\begin{align}\label{sum_neg_beta}
Im\left( \sum_{j=1}^{N} \log \langle\boldsymbol{v}\left(\boldsymbol{\kappa}_j\right), \boldsymbol{v}\left(\boldsymbol{\kappa}_{j+1}\right)\rangle\right) \in \left(\pi,3\pi\right).
\end{align}
The statement $(iii)$ follows by (\ref{discrete_chern_formula}), (\ref{sum_pos_beta}) and (\ref{sum_neg_beta}).
\end{proof}

\subsection{Edge Modes for the Topological Mechanical System}\label{interface_mode}

\begin{figure}[h]
\centering
\includegraphics[scale=0.4]{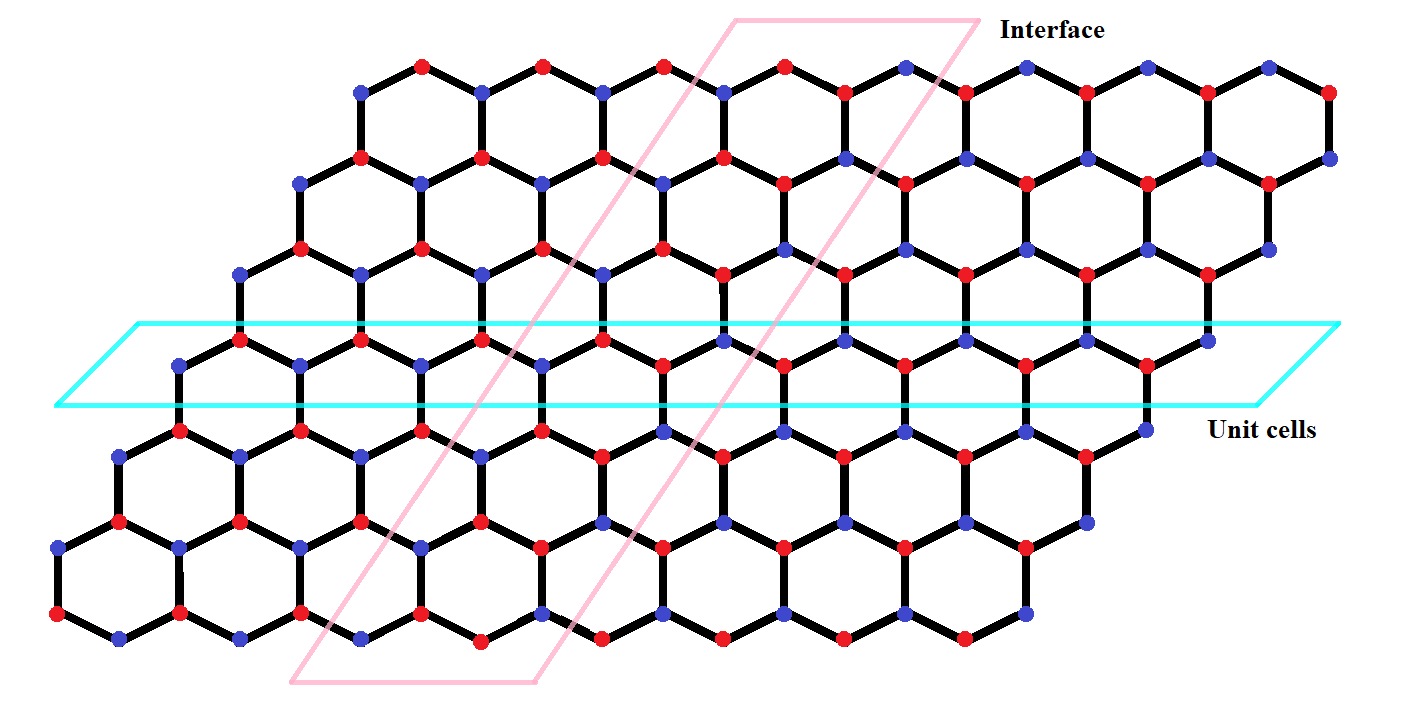}
\caption{The topological mechanical system formed by two hexagonal lattices with opposite $\beta$ values. The interface direction is parallel to $\mathbf{a}_2$.}
\label{gencase_2d_interface}
\end{figure}

We consider a joint system formed by gluing two periodic hexagonal lattices with opposite $\beta$ values. It forms an interface parallel to $\mathbf{a}_2$ where two identical masses are connected as shown in Figure \ref{gencase_2d_interface}. For $p\in \mathbb{Z}^+$ and $q \in \mathbb{Z} $, the governing equations for the displacements of the masses located at the left and right side of the interface are
\begin{equation} \label{eqn_2d_left}
\begin{aligned}
m_a \left(U_{-p,q}^a\right)'' + k\left(U_{-p,q}^a - U_{-p,q}^b\right) - k\left(U_{-p,q}^a - U_{-(p-1),q}^b\right) - k\left(U_{-p,q}^a - U_{-p,q+1}^b\right)&=0,\\
m_b \left(U_{-(p-1),q}^b\right)'' + k\left(U_{-(p-1),q}^b - U_{-(p-1),q-1}^a\right) - k\left(U_{-(p-1),q}^b - U_{-(p-1),q}^a\right) - k\left(U_{-(p-1),q}^b - U_{-p,q}^a\right)&=0,\\
\end{aligned}
\end{equation}
and
\begin{equation} \label{eqn_2d_right}
\begin{aligned}
m_a \left(U_{p,q}^a\right)'' + k\left(U_{p,q}^a - U_{p-1,q}^b\right)+ k\left(U_{p,q}^a - U_{p,q-1}^b\right)+ k\left(U_{p,q}^a - U_{p,q}^b\right) &= 0, \\
m_b \left(U_{p-1,q}^b\right)'' + k\left(U_{p-1,q}^b - U_{p-1,q}^a\right)+ k\left(U_{p-1,q}^b - U_{p,q}^a\right)+ k\left(U_{p-1,q}^b - U_{p-1,q+1}^a\right) &= 0,
\end{aligned}
\end{equation}
respectively. The governing equations for the displacement of the masses located at the interface are
\begin{equation} \label{eqn_2d_middle}
\begin{aligned}
m_a \left(U_{1,q}^a\right)'' + k \left(U_{1,q}^a - U_{-1,q}^a\right) + k \left(U_{1,q}^a - U_{1,q-1}^b\right) + k\left(U_{1,q}^a - U_{1,q}^b\right) &= 0,\\
m_a \left(U_{-1,q}^a\right)'' + k \left(U_{-1,q}^a - U_{1,q}^a\right) + k \left(U_{-1,q}^a - U_{-1,q+1}^b\right) + k\left(U_{-1,q}^a - U_{-1,q}^b\right) &= 0.
\end{aligned}
\end{equation}

For each $k_{||} \in [0,\frac{2\pi}{a}]$, we consider the solutions that propagate along the interface and decay along the horizontal direction by letting
\begin{align}\label{general_soln}
U_{p,q}^a(t) = u^a_p e^{i\omega\gamma}e^{ik_{||}q}\,\,\, \text{and}\,\,\, U_{p,q}^b(t) = u^b_p e^{i\omega\gamma}e^{ik_{||}q},
\end{align}
where $\gamma = \sqrt{k/m}t$. An edge mode $(U_{p,q}^a(t),U_{p,q}^b(t))$ is the nontrivial solution to (\ref{eqn_2d_left}) - (\ref{eqn_2d_middle}) which decays to zero as $p \to \infty$.

We aim to show that there exist edge modes for the joint system in Figure \ref{gencase_2d_interface} when $\beta\neq 0$. For each $k_{||} \in [0,\frac{2\pi}{a}]$, such edge mode attains an eigenfrequency $\omega\left(k_{||}\right)$ such that $\omega^2\left(k_{||}\right)$ located in the band gap $(\lambda_-\left(k_{||}\right),\lambda_+\left(k_{||}\right))$, where
\begin{align*}
\lambda_-\left(k_{||}\right) = \max_{\tilde{\boldsymbol{\kappa}}}\lambda_{-}\left(\tilde{\boldsymbol{\kappa}}\right)\,\,\,\text{and}\,\,\,\lambda_+\left(k_{||}\right) = \min_{\tilde{\boldsymbol{\kappa}}}\lambda_{+}\left(\tilde{\boldsymbol{\kappa}}\right)).
\end{align*}
In the above, $\tilde{\boldsymbol{\kappa}} = k_1 \boldsymbol{b}_1 + \frac{k_{||}a^2}{2\pi}\boldsymbol{b}_2$ and $k_1 \in (0,1)$.
\begin{remark}\label{opp_sign_beta}
For $k_{||} \in [0,\frac{2\pi}{a}]$, $\lambda_-\left(k_{||}\right)$ and $\lambda_+\left(k_{||}\right)$ occur when
\begin{align*}
\tilde{\boldsymbol{\kappa}} = \left(\frac{n}{2}+\frac{k_{||}a^2}{4\pi}\right)\boldsymbol{b}_1 + \frac{ k_{||}a^2}{2\pi}\boldsymbol{b}_2,\,\,\,\text{for}\,\,\,n \in \{0,1\},
\end{align*}
with $(-1)^n\cos\left(\frac{k_{||}a^2}{2}\right) < 0$ and we have
\begin{align}\label{min_max_eval}
\lambda_-\left(k_{||}\right) = \frac{3-\sqrt{9\beta^2 + \left(1-\beta^2\right)|d\left(\tilde{\boldsymbol{\kappa}}\right)|^2}}{1-\beta^2}\,\,\,\text{and}\,\,\,\lambda_+\left(k_{||}\right) = \frac{3+\sqrt{9\beta^2 + \left(1-\beta^2\right)|d\left(\tilde{\boldsymbol{\kappa}}\right)|^2}}{1-\beta^2},
\end{align}
where $\left|d\left(\tilde{\boldsymbol{\kappa}}\right)\right| = \left|1+\left(-1\right)^n2\cos\left(\frac{k_{||}a^2}{2}\right)\right|$. It is clear that $\lambda_-\left(k_{||}\right) < \lambda_+\left(k_{||}\right)$.
\end{remark}
Our main result is stated in the following theorem.
\begin{theorem}\label{two_d_main_thm}(Existence of edge mode)
If two periodic systems with opposite $\beta$ values are connected as in Figure \ref{cell_numb}, then for each $k_{||} \in [0,\frac{2\pi}{a}]$, there exists an edge mode $(U_{p,q}^a(t),U_{p,q}^b(t))$ in the form (\ref{general_soln}) with $\omega^2\left(k_{||}\right) \in (\lambda_{-}\left(k_{||}\right), \lambda_{+}\left(k_{||}\right))$.
\end{theorem}
\begin{remark}
By Lemma \ref{chern_lemma}, if $\beta_1$ and $\beta_2$ attain opposite signs, then $C_{K_j}^{(1)}C_{K_j}^{(2)}<0$ where $C_{K_j}^{(i)}$ is the valley Chern number calculated with $\beta_i$ at $K_j$.
\end{remark}

\subsubsection{Transfer Matrix for the Periodic System}

In this subsection, we compute the transfer matrices for the lattices on the left and right side of the interface. To simplify the calculations, we introduce the following notations:
\begin{align*}
\tau\left(k_{||}\right) &:= 3-(1+\beta)\omega^2\left(k_{||}\right),\\
\sigma\left(k_{||}\right) &:= 3 - (1-\beta)\omega^2\left(k_{||}\right),\\
z\left(k_{||}\right) &:= 1+e^{ik_{||}a^2},\\
\xi\left(\sigma,k_{||}\right) &:= \tau\sigma-1-|z|^2,
\end{align*}
We consider the solutions in the form of (\ref{general_soln}) for (\ref{eqn_2d_left})-(\ref{eqn_2d_middle}). For the right periodic system, by (\ref{eqn_2d_right}), $u^a_p$ and $u^b_p$ satisfy
\begin{equation}\label{2d_right_1}
\begin{aligned}
\tau u_p^a - u_{p-1}^b - u_{p}^b - u_{p}^be^{-ik_{||}a^2}&=0,\\
\sigma u_{p-1}^b - u_{p-1}^a - u_{p}^a - u_{p-1}^ae^{ik_{||}a^2}&=0,
\end{aligned}
\end{equation}
which is reduced to $A\left(k_{||}\right)\mathbf{u}_{p-1}=B\left(k_{||}\right)\mathbf{u}_p$ where
\begin{align*}
A\left(k_{||}\right) = \begin{bmatrix}
0 & 1 \\
-z\left(k_{||}\right) & \sigma
\end{bmatrix}
\,\,\,, 
B\left(k_{||}\right) = \begin{bmatrix}
\tau & -\overline{z\left(k_{||}\right)}\\
1 & 0
\end{bmatrix}
\,\,\,\text{and}\,\,\,
\mathbf{u}_p = \begin{bmatrix}
u_p^a\\
u_p^b
\end{bmatrix} 
= \begin{bmatrix}
u^a\\
u^b
\end{bmatrix}_p.
\end{align*}
We rewrite $A\mathbf{u}_{p-1}=B\mathbf{u}_p$ as
\begin{align*}
T_R\left(k_{||}\right)\mathbf{u}_{p-1} = \mathbf{u}_p,
\end{align*}
where the transfer matrix
\begin{align*}
T_R\left(k_{||}\right)= B^{-1}A = \frac{1}{\overline{z}}\begin{bmatrix}
-|z|^2 & \overline{z}\sigma\\
-z\tau & \tau\sigma-1
\end{bmatrix}.
\end{align*}
The eigenpairs of $T_R$ are 
\begin{align*}
\lambda_{R,\pm}\left(k_{||}\right) = \frac{1}{2\overline{z}}\left[\xi \pm \sqrt{\xi^2 - 4|z|^2}\right]\,\,\,\text{and}\,\,\,
\mathbf{v}_{R,\pm}\left(k_{||}\right) = \begin{bmatrix}
\sigma\\
\lambda_{R,\pm}+z
\end{bmatrix}.
\end{align*}
Similarly, by (\ref{eqn_2d_left}), $u^a_{-p}$ and $u^b_{-p}$ satisfy
\begin{align*}
\tau u_{-p}^a - u_{-p}^b - u_{-(p-1)}^b - e^{ik_{||}a^2}u_{-p}^b&=0,\\
\tau u_{-(p-1)}^b - u_{-(p-1)}^a - u_{-(p-1)}^a - e^{-ik_{||}a^2}u_{-p}^a&=0.
\end{align*}
We obtain $T_L\left(k_{||}\right) \mathbf{u}_{-p} = \mathbf{u}_{-(p+1)}$, where the transfer matrix
\begin{align*}
T_L\left(k_{||}\right) = \frac{1}{z}\begin{bmatrix}
-|z|^2 & z\sigma\\
-\overline{z}\tau & \tau\sigma-1
\end{bmatrix}.
\end{align*}
Since $T_R = \overline{T_L}$, the eigenpairs of $T_L$ are
\begin{align*}
\lambda_{L,\pm}\left(k_{||}\right) = \overline{\lambda_{R,\pm}\left(k_{||}\right)} = \frac{1}{2z}\left[\xi \pm \sqrt{\xi^2 - 4|z|^2}\right]\,\,\,\text{and}\,\,\,
\mathbf{v}_{L,\pm}\left(k_{||}\right) = \overline{\mathbf{v}_{R,\pm}\left(k_{||}\right)} = \begin{bmatrix}
\sigma\\
\lambda_{L,\pm}+\overline{z}
\end{bmatrix}.
\end{align*}
\begin{remark}\label{z_d_rel}
Along the interface, we have
\begin{align*}
|z|^2 = \left|1+ e^{ik_{||}a^2}\right|^2 = \left((-1)^{n+1}2\cos\left(\frac{k_{||}a^2}{2}\right)\right)^2,
\end{align*}
and $\left|d\right|^2 = \left(1-\left|z\right|\right)^2$.
\end{remark}
In order for $\mathbf{u}_p$ to decay as $p \to \pm\infty$, it is necessary that $\mathbf{u}_1 = \begin{bmatrix}
u^a\\
u^b
\end{bmatrix}_1$ and $\mathbf{u}_{-1} = \begin{bmatrix}
u^a\\
u^b
\end{bmatrix}_{-1}$ are parallel to the eigenvectors of $T_R$ and $T_L$ whose corresponding eigenvalues have absolute value less than 1.
\begin{remark}\label{evec_xi}
(i) If $\xi^2 - 4|z|^2 < 0$, then
\begin{align*}
\left|\lambda_R\right|^2 =\left|\lambda_L\right|^2 = \frac{1}{4|z|^2}\left|\xi \pm i \sqrt{4 |z|^2 - \xi^2}\right|^2 = 1.
\end{align*}
Since we consider $\left|\lambda_{L,R}\right|<1$, there holds $\xi^2 - 4|z|^2\geq 0$.\\
(ii) If $\left|\lambda_{R,+}\right|^2 < 1$, then
\begin{align*}
\xi\sqrt{\xi^2 - 4\left|z\right|^2} < 0,
\end{align*}
which implies $\xi< 0$. If $\left|\lambda_{R,-}\right|^2 < 1$, then
\begin{align*}
-\xi\sqrt{\xi^2 - 4\left|z\right|^2} < 0,
\end{align*}
which implies $\xi > 0$. 
\end{remark}
By Remark \ref{evec_xi}, the eigenvalue $\lambda_R$ with $|\lambda_R|<1$ is
\begin{align}\label{evalues}
\lambda_R = 
\begin{cases}
\frac{1}{2\overline{z}}\left(\xi - \sqrt{\xi^2 - 4\left|z\right|^2}\right),\,\,\,\text{for}\,\,\,\xi > 0,\\
\frac{1}{2\overline{z}}\left(\xi + \sqrt{\xi^2 - 4\left|z\right|^2}\right),\,\,\,\text{for}\,\,\,\xi < 0.
\end{cases} 
\end{align}
Since $\lambda_{L,\pm} = \overline{\lambda_{R,\pm}}$, we have 
\begin{align}\label{evalues_L}
\lambda_L = 
\begin{cases}
\frac{1}{2z}\left(\xi - \sqrt{\xi^2 - 4\left|z\right|^2}\right),\,\,\,\text{for}\,\,\,\xi > 0,\\
\frac{1}{2z}\left(\xi + \sqrt{\xi^2 - 4\left|z\right|^2}\right),\,\,\,\text{for}\,\,\,\xi < 0.
\end{cases} 
\end{align}
Thus the parallelism condition above implies that
\begin{align}\label{parallelvec}
\begin{bmatrix}
u_1^a\\
u_1^b
\end{bmatrix}
= c_1 \mathbf{v}_R
=
\begin{bmatrix}
c_1 \sigma\\
c_1 \left(\lambda_R+z\right)
\end{bmatrix}
\,\,\,\text{and} \,\,\,
\begin{bmatrix}
u_{-1}^a\\
u_{-1}^b
\end{bmatrix}
= c_2\mathbf{v}_L
=
\begin{bmatrix}
c_2\sigma\\
c_2\left(\lambda_L+\overline{z}\right)
\end{bmatrix}.
\end{align}

\subsubsection{Proof of Theorem~\ref{two_d_main_thm}}

By (\ref{eqn_2d_middle}) and (\ref{general_soln}), we obtain
\begin{align*}
m_a (-\omega^2)\frac{k}{m}u_1^a e^{i\omega\tau} e^{ik_{||}q} + k\left[u_1^a- u_{-1}^b\right]e^{i\omega\tau} e^{ik_{||}q}  + k\left[u_1^a - u_1^b\right]e^{i\omega\tau} e^{ik_{||}q} + k \left[u_1^a - u_1^be^{-ik_{||}}\right]e^{i\omega\tau} e^{ik_{||}q} &= 0,\\
m_a(-\omega^2)\frac{k}{m}u_{-1}^ae^{i\omega\tau} e^{ik_{||}q} + k\left[u_{-1}^a - u_1^a\right]e^{i\omega\tau} e^{ik_{||}q} + k\left[u_{-1}^a - u_{-1}^b\right]e^{i\omega\tau} e^{ik_{||}q} + k\left[u_{-1}^a - u_{-1}^be^{ik_{||}}\right]e^{i\omega\tau} e^{ik_{||}q} &= 0,
\end{align*}
which implies
\begin{equation}
\begin{aligned}\label{2d_middle_eqn}
\tau u_{1}^a - u_{-1}^a  -\overline{z}u_{1}^b&=0,\\
\tau u_{-1}^a - u_{1}^a -zu_{-1}^b&=0.
\end{aligned}
\end{equation}
Then (\ref{2d_middle_eqn}) can be simplified as
\begin{equation}
\begin{aligned} \label{c1c2}
\tau c_1\sigma - \overline{z} c_1 \left(\lambda_R z\right) - c_2 \sigma&=0,\\
\tau c_2\sigma - z c_2 \left(\lambda_L \overline{z}\right) - c_1 \sigma&=0.
\end{aligned}
\end{equation}
For $\sigma = 0$, the first equation in (\ref{c1c2}) implies that 
\begin{align*}
c_1\lambda_R|z|^2 = 0.
\end{align*}
which is equivalent to $c_1=0$ and it gives the trivial solution for (\ref{2d_right_1}). The first equation in (\ref{c1c2}) implies that, for $\sigma\neq0$, 
\begin{align}\label{firsteqn}
c_2 = \frac{1}{\sigma} c_1 \left(\tau\sigma - \overline{z}\left(\lambda_R+z\right)\right),
\end{align}
(\ref{firsteqn}) and the second equation in (\ref{c1c2}) together imply that
\begin{align*}
 c_1\sigma &= \frac{1}{\sigma}c_1\left(\tau\sigma - \overline{z}\left(\lambda_R+z\right)\right)\left(\tau\sigma - z\left(\lambda_L+\overline{z}\right)\right),
\end{align*}
which is equivalent to
\begin{equation}
\begin{aligned}\label{maineqn}
4\sigma^2 = 
\begin{cases}
\left|\xi + 2 - \sqrt{\xi^2 - 4\left|z\right|^2}\right|^2,\,\,\,\text{for}\,\,\,\xi < 0,\\
\left|\xi + 2 + \sqrt{\xi^2 - 4\left|z\right|^2}\right|^2,\,\,\,\text{for}\,\,\,\xi > 0,
\end{cases}
\end{aligned}
\end{equation}
by (\ref{evalues}). We consider $\xi < 0$ and $\xi > 0$ respectively to obtain a solution to (\ref{maineqn}). We have, if $\xi < 0$,
\begin{align*}
2\sigma = \xi+2-\sqrt{\xi^2 - 4\left|z\right|^2}\,\,\,\text{or}\,\,\,2\sigma = -\left(\xi+2\right)+\sqrt{\xi^2 - 4\left|z\right|^2}.
\end{align*}
$(i)$ If $2\sigma = \xi+2-\sqrt{\xi^2 - 4\left|z\right|^2}$, then
\begin{align*}
4\sigma^2 - 4\sigma\left(\xi+2\right) + \xi^2 + 4\xi+4 = \xi^2 - 4\left|z\right|^2,
\end{align*}
which implies
\begin{align}\label{lasteqn}
- \frac{1+\beta}{1-\beta}\sigma^3 + \left[1+\frac{6\beta}{1-\beta} + \frac{6\beta}{1-\beta} \right] \sigma^2 -\sigma\left[1-\left|z\right|^2 + \frac{6\beta}{1-\beta} \right]= 0.
\end{align}
The above equation attains two nonzero roots:
\begin{align*}
\sigma_1 = \frac{1+3\beta - \sqrt{4\beta^2 + \left(1-\beta^2\right)\left|z\right|^2}}{\left(1+\beta\right)},\,\,\, \sigma_2 = \frac{1+3\beta + \sqrt{4\beta^2 + \left(1-\beta^2\right)\left|z\right|^2}}{\left(1+\beta\right)}.
\end{align*}
Note that, for $\beta \in (-1,1)$,
\begin{align*}
0 > \xi(\sigma_2,k_{||}) \geq 4 \frac{\left|\beta\right| -\beta^2}{1-\beta^2} > 0,
\end{align*}
which is a contradiction to $\xi<0$. For $\sigma_1$,
\begin{align*}
\xi\left(\sigma_1,k_{||}\right) = \frac{1+\beta}{1-\beta}\sigma^2 - \frac{6\beta}{1-\beta} \sigma - \left|z\right|^2 - 1 = -\frac{4\beta^2 + 2\sqrt{4\beta^2 + \left(1-\beta^2\right)|z|^2}}{1-\beta^2} < 0.
\end{align*}
Therefore,
\begin{align*}
\tilde\sigma_1 = \frac{1+3\beta - \sqrt{4\beta^2 + \left(1-\beta^2\right)\left|z\right|^2}}{\left(1+\beta\right)}
\end{align*}
is one root of (\ref{lasteqn}).\\
$(ii)$ If $2\sigma = -\left(\xi+2\right)+\sqrt{\xi^2 - 4\left|z\right|^2}$, by similar calculations, we obtain
two more roots
\begin{align*}
\sigma_3 = \frac{-1+3\beta + \sqrt{16\beta^2 + \left(1-\beta^2\right)\left|z\right|^2}}{1+\beta},\,\,\,\text{and} \,\,\,\sigma_4 = \frac{-1+3\beta - \sqrt{16\beta^2 + \left(1-\beta^2\right)\left|z\right|^2}}{1+\beta}.
\end{align*}
Similarly, we have
\begin{align*}
0 &> \xi\left(\sigma_4,k_{||}\right) = 2\frac{4\beta^2 + 4\left|\beta\right|}{1-\beta^2}\geq 0,
\end{align*}
\begin{align*}
\xi\left(\sigma_3,k_{||}\right) = \frac{8\beta^2 - 2\sqrt{16\beta^2 + \left(1-\beta^2\right)\left|z\right|^2}}{1-\beta^2} \leq 8\frac{\beta^2-|\beta|}{1-\beta^2} \leq 0.
\end{align*}
Thus, we obtain
\begin{align*}
\tilde\sigma_2 = \frac{-1+3\beta + \sqrt{16\beta^2 + \left(1-\beta^2\right)\left|z\right|^2}}{1+\beta},
\end{align*}
as another root of (\ref{maineqn}).

By similar calculations for $\xi>0$, we obtain two more roots:
\begin{align*}
\tilde\sigma_3 = \frac{1+3\beta+\sqrt{4\beta^2 + \left(1-\beta^2\right)|z|^2}}{1+\beta}\,\,\,\text{and}\,\,\,\tilde\sigma_4 = \frac{-1+3\beta - \sqrt{16\beta^2 + \left(1-\beta^2\right)|z|^2}}{1+\beta}.
\end{align*}

From the relation $\sigma = 3 - \left(1+\beta\right)\omega^2$, we obtain the corresponding eigenvalues: 
\begin{equation}
\begin{aligned}\label{omegas_2d}
\omega_1^2\left(k_{||}\right) &=  \frac{2+\sqrt{4\beta^2 + \left(1-\beta^2\right)\left|z\right|^2}}{1-\beta^2},\\
\omega_2^2\left(k_{||}\right) &= \frac{4-\sqrt{16\beta^2 + \left(1-\beta^2\right)\left|z\right|^2}}{1-\beta^2},\\
\omega_3^2\left(k_{||}\right) &= \frac{2-\sqrt{4\beta^2 + \left(1-\beta^2\right)\left|z\right|^2}}{1-\beta^2},\\
\omega_4^2\left(k_{||}\right) &= \frac{4+\sqrt{16\beta^2 + \left(1-\beta^2\right)\left|z\right|^2}}{1-\beta^2}.
\end{aligned}
\end{equation}

Next we show that $\omega_j^2\left(k_{||}\right) \in \left(\lambda_-\left(k_{||}\right),\lambda_+\left(k_{||}\right)\right)$ for $j = 1,2$ but $\omega_j^2\left(k_{||}\right) \notin \left(\lambda_-\left(k_{||}\right),\lambda_+\left(k_{||}\right)\right)$ for $j=3,4$. For $\omega_1^2\left(k_{||}\right)$, we have
\begin{align*}
w_1^2 - \lambda_- &=\frac{2+\sqrt{4\beta^2 + \left(1-\beta^2\right)|z|^2}}{1-\beta^2}-\frac{3 - \sqrt{9\beta^2 + \left(1-\beta^2\right)|d|^2}}{1-\beta^2}\\
&\geq \frac{1}{1-\beta^2}\left[-1+\sqrt{1+2\left[4\beta^2 + \left(1-\beta^2\right)|z|^2\right]}\right] \geq 0.
\end{align*}
In the above, we have used Remark (\ref{z_d_rel}) to relate $z$ and $d$. Thus, $\omega_1^2 \geq \lambda_-$. Similarly,
\begin{align*}
w_1^2 - \lambda_+ &= \frac{2+\sqrt{4\beta^2 + \left(1-\beta^2\right)|z|^2}}{1-\beta^2}-\frac{3 + \sqrt{9\beta^2 + \left(1-\beta^2\right)|d|^2}}{1-\beta^2}\\
&\leq \frac{1}{1-\beta^2}\left[\sqrt{4\beta^2+\left(1-\beta^2\right)|z|^2} - 1 - \sqrt{4\beta^2+\left(1-\beta^2\right)|z|^2 + 4}\right] \leq 0.
\end{align*}
Thus, $\omega_1^2 \leq \lambda_+$ and we have $\omega_1^2 \in \left[\lambda_-,\lambda_+\right]$. By similar calculations, we can show $\omega_2^2 \in \left[\lambda_-,\lambda_+\right]$. However, if $\omega_3^2 \geq \lambda_-$, then
\begin{align*}
&\frac{2-\sqrt{4\beta^2 + \left(1-\beta^2\right)\left|z\right|^2}}{1-\beta^2} \geq \frac{3 - \sqrt{9\beta^2 + \left(1-\beta^2\right)\left|d\right|^2}}{1-\beta^2},
\end{align*}
which can be simplified as
\begin{align*}
 -\beta^2 \left(1-\beta^2\right) \left(\left|z\right| - 2\right)^2 \geq 0,
\end{align*}
which is impossible. Thus, $\omega_3^2 < \lambda_-$ and $\omega_3^2 \notin \left[\lambda_-,\lambda_+\right]$. Similarly, we can show $\omega_4^2 > \lambda_+$ and $\omega_4^2 \notin \left[\lambda_-,\lambda_+\right]$.
Figure \ref{bands_omegas} shows that $\omega_1^2\left(k_{||}\right)$ and $\omega_2^2\left(k_{||}\right)$ located in the band gap $\left(\lambda_-\left(k_{||}\right),\lambda_+\left(k_{||}\right)\right)$ when $a = 1$ and $\beta=0.1$.
\begin{figure}[h]
\centering
\includegraphics[scale=0.3]{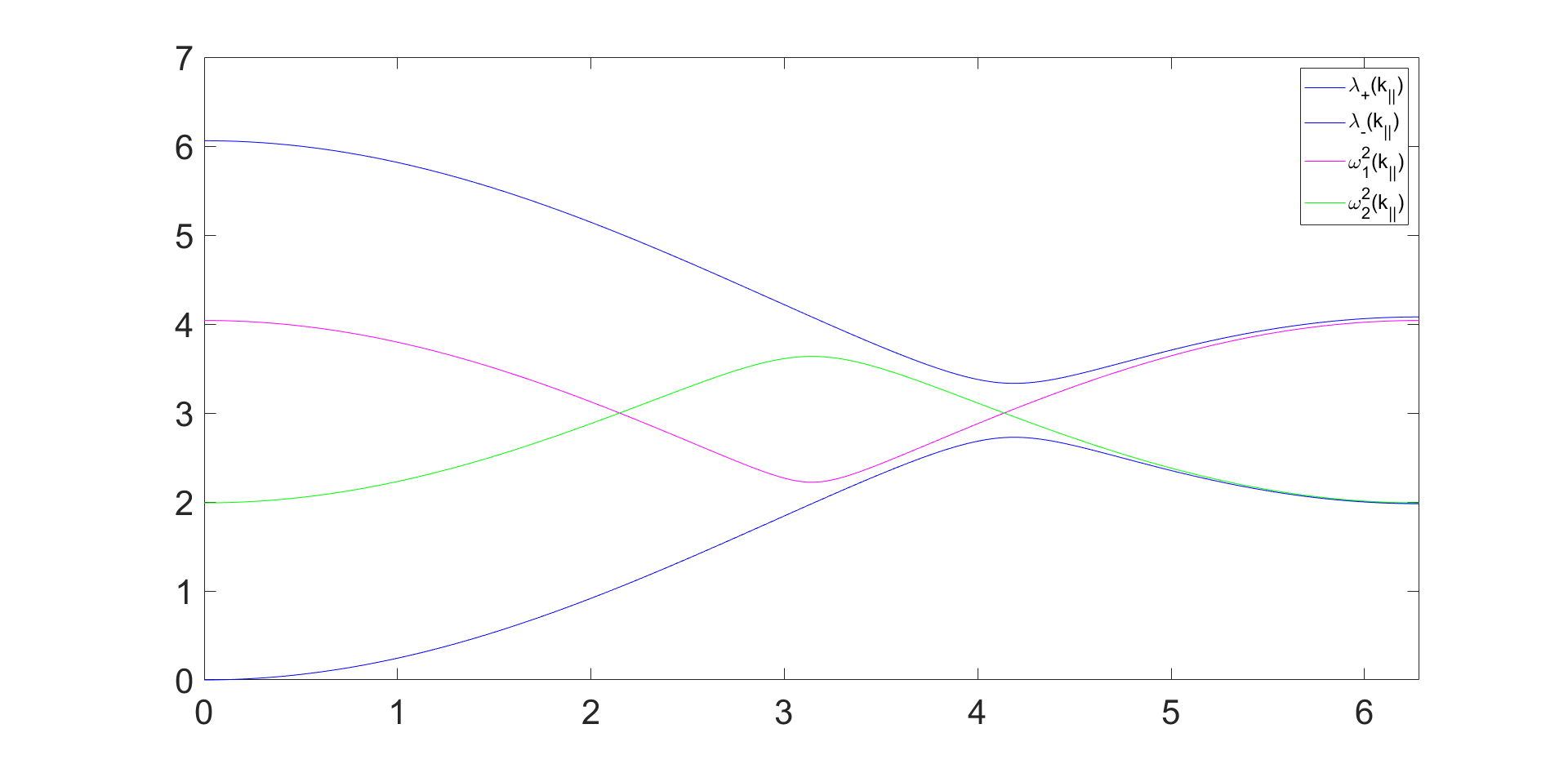}
\caption{The upper and lower bands $\lambda_\pm\left(k_{||}\right)$, and the eigenvalues of the edge modes, $\omega_1^2\left(k_{||}\right)$ and $\omega_2^2\left(k_{||}\right)$, when $a=1$ and $\beta = 0.1$.}
\label{bands_omegas}
\end{figure}

\section*{Acknowledgements} The work of R. Ozdemir and J. Lin is partially supported the NSF grant DMS-2011148.
\pagebreak

\bibliography{references.bib}
\bibliographystyle{ieeetr}

\end{document}